\documentclass[12pt, reqno, a4paper,pagebackref]{amsart}
\usepackage{graphicx}
\usepackage{fancyhdr}
\usepackage{enumerate}
\usepackage{amsmath}

\usepackage[multiple]{footmisc}

\usepackage{amsthm}
\usepackage{mathabx}
\usepackage{amssymb}
\usepackage{relsize}
\usepackage{pdfsync}
\usepackage[colorlinks=true,linkcolor=blue,citecolor=magenta]{hyperref}

\usepackage[normalem]{ulem}

\usepackage{xcolor}

\usepackage{cancel}

\usepackage{calligra}

\definecolor{darkgreen}{rgb}{0,0.5,0}

\usepackage{tikz}

\usepackage[bottom,first]{draftcopy}

\usepackage{changebar}
\usepackage{pdfsync}
\usetikzlibrary{decorations.pathmorphing,patterns}
\usetikzlibrary{calc,arrows}

\usepackage{pgf}

\usepackage{hyperref}

\usetikzlibrary{automata}
\usetikzlibrary{positioning}

\tikzset{
    state/.style={
           rectangle,
           rounded corners,
           draw=black, very thick,
           minimum height=2em,
           inner sep=2pt,
           text centered,
           },
}

\usepackage{MnSymbol}

\usepackage{changebar}
\usepackage{pdfsync}

\usepackage[normalem]{ulem}

\usepackage{mathtools}

\makeindex
\newtheorem{theorem}{Theorem}[section]
\newtheorem{lemma}[theorem]{Lemma}
\newtheorem{proposition}[theorem]{Proposition}
\newtheorem{corollary}[theorem]{Corollary}

\newtheorem{remark}[theorem]{Remark}

\newcommand{\cal}{\mathcal}

\newcommand\bbR{{\mathbb R}}
\newcommand\bbZ{{\mathbb Z}}

\newcommand\al{\alpha}

\newcommand\lang{\langle\langle}
\newcommand\rang{\rangle\rangle}

\newcommand\bbE{{\mathbb E}}

\newcommand\R{{\mathbb R}}

\newcommand\bbP{{\mathbb P}}

\newcommand \ga{\gamma}
\newcommand \om{\omega}
\newcommand \la{\lambda}

\newcommand{\bbT}{\mathbb T}
\newcommand{\ii}{\mathrm{i}}

\renewcommand{\ge}{\geqslant}
\renewcommand{\le}{\leqslant}
\newcommand{\dd  }{\mathrm{d}}

\renewcommand{\tilde}{\widetilde}
\renewcommand{\bar}{\overline}
\numberwithin{equation}{section}

\newcommand{\rv}{{\bf r}}
\newcommand{\pv}{{\bf p}}
\newcommand{\qv}{{\bf q}}

\newcommand{\cF}{{\mathcal F}} 

\newcommand{\jl}[1]{{\color{darkgreen}#1}}

\newcommand{\red}{\textcolor{red}}

\begin{document}

\title[Forced harmonic chain]{Heat flow in
   a periodically forced, thermostatted chain}

\author{Tomasz Komorowski}
 \address{Tomasz Komorowski\\Institute of Mathematics,
   Polish Academy
  Of Sciences\\Warsaw, Poland.} 
\email{{\tt tkomorowski@impan.pl}}

\author{Joel L. Lebowitz}
\address{Joel Lebowitz, Departments of Mathematics and Physics,  Rutgers University}
\email{lebowitz@math.rutgers.edu}
 
 \author{Stefano Olla}
 \address{Stefano Olla, CEREMADE,
   Universit\'e Paris Dauphine - PSL Research University \\
\emph{and}  Institut Universitaire de France\\
\emph{and} GSSI, L'Aquila}
  \email{olla@ceremade.dauphine.fr}

\date{\today {\bf File: {\jobname}.tex.}}

\begin{abstract}
{We investigate the properties of a harmonic chain in contact with a
  thermal bath at one end  and subjected, at its other end, to a periodic force.
The particles {also} {undergo} a random velocity reversal action,
{which results in} a finite heat conductivity {of} 
 the system.
We prove the approach of the system to a time periodic state and compute the
heat current, equal to the time averaged work done on the system, in that state.
This work approaches a finite positive value as the length of the
chain increases. Rescaling space, the strength {and/or the period of} the force  leads
to a macroscopic temperature profile corresponding to the stationary
solution of a continuum heat equation with   {Dirichlet-Neumann}
boundary conditions.}
\end{abstract}

\thanks{} 

\maketitle


\section{Introduction}

The conversion of mechanical energy to heat, accompanied by
the production of entropy,  is a very common phenomenon in nature.
It occurs on all scales 
     of {space} and time: from ocean tides to the movement of a charged
     particle in a fluid under the influence of an electric field.
     {It happens} each time we rub our hands or scratch our head.
     A fully microscopic description of the phenomenon is
     desirable but very complicated.
     Here we study an example of this phenomenon in a very simple
     microscopic model system. In particular, we consider
     a linear chain (system) of $n+1$ particles,
{labelled by $x=0,\ldots,n$,} in contact with a thermal reservoir at one end and
acted {upon} by a periodic force at its other end.

The interaction of the system with the reservoir
is stochastic, modeled as usual by an Ornstein-Uhlenbeck process (the Langevin force).
The action
of the external force on the other hand,
is deterministic described by a time periodic Hamiltonian.
There is also a stochastic velocity flip acting on the particles of the system
to model non linear interactions, {inducing a diffusive behaviour of the energy}
and {producing} a finite heat conductivity.

  We prove the existence of a unique
periodic state approached, as $t\to+\infty$, from any initial
state. The average work done on the system over a period is equal to the time averaged heat flux, $J_n$, going
into the reservoir. This is computed explicitly for all $n$.
{\jl{The} diffusive behaviour of the energy implies that the heat flux $J_n$
  is proportional to the microscopic gradient ot the temperature.}
To maintain a spatially constant time averaged macroscopic  current $J$
in the limit   $n\to\infty$ one needs to have a spatially constant time averaged
   temperature gradient proportional to $J$.
 This can be achieved  by appropriately scaling the force and/or period
$\theta_n$ as a function of $n$.  This leads to
    a {macroscopic} spatial temperature profile $T(u)$, where 
$u=x/n \in [0,1]$ is the scaled spatial coordinate,  given  by the stationary solution of the heat equation with a fixed temperature $T(0)=T_-$ (the temperature
  of the heat reservoir {that is placed at the left endpoint of the
    chain}) and a fixed energy current $J$ entering the system on the
  right, $u=1$ ({where the periodic force is applied}). The  thermal
  conductivity $\kappa$ in the heat equation {can be} computed explicitly by the Kubo
  formula for this system (see \cite{BKLL11}
  as well as the comments in \autoref{sec:conclusions}).
  It is independent of the temperature, so, {as
    a result}, $T(u)$ is linear in $u$, {with $\kappa T'(u) = - J$.
{    As already mentioned above, it is the presence of the stochastic flip in the bulk
    that is responsible for the conversion of the work done by the periodic force into heat
    that is diffusively transported through the system. Detailed study of the thermalization
    in the bulk for such stochastic dynamics has been studied in \cite{lukk}.
      }

  {We note that our setting differs from  the typical setup, in which the \emph{stationary}
  macroscopic energy transport}  has been studied. {E.g. in
\cite{RLL67,kos2,bll, bo1} the chain  is placed  between
  heat baths at different temperatures.}
  The stationary temperature profile depends then on the
  boundary temperatures. As far as we know,
  the present article is the first {to derive} a  rigorous macroscopic limit for a system
  in a periodic state induced by a periodic external force.
  Periodic states of a system under external periodic forcing {have
    been previously} considered in
  \cite{joel57}.
}

\subsection{Description of the model}
The configuration of particle positions and momenta are described by
\begin{equation}
  \label{eq:1}
  (\mathbf q, \mathbf p) =
  (q_0, \dots, q_n, p_0, \dots, p_n) \in \R^{n+1}\times\R^{n+1}. 
\end{equation}
We should think of the positions $q_x$ as relative displacement from a point,
say $x$ in a finite lattice $\{0,1,\ldots,n\}$.
The total energy of the chain is defined by the Hamiltonian:
$\mathcal{H}_n (\mathbf q, \mathbf p):=
\sum_{x=0}^n {\cal E}_x (\mathbf q, \mathbf p),$
where the microscopic energy is given by
\begin{equation}
\label{Ex}
{\cal E}_x (\mathbf q, \mathbf p):=  \frac{p_x^2}2 +
\frac12 (q_{x}-q_{x-1})^2 +\frac{\om_0^2 q_x^2}{2} ,\
\quad x = 0, \dots, n,
\end{equation}
{with the pinning constant $\om_0>0$. }
{We adopt the convention that $q_{-1}:=q_0$.}


The  microscopic dynamics of
the process $\{(\mathbf q(t), \mathbf p(t))\}_{t\ge0}$
describing the total chain is  given in the bulk by

\begin{equation} 
\label{eq:flip}
\begin{aligned}
  \dot   q_x(t) &= p_x(t) ,
  \qquad \qquad \qquad  \qquad \qquad x\in \{0, \dots, n\},\\
  \dd   p_x(t) &=  \left(\Delta_N q_x-\om_0^2 q_x\right) \dd t-   2 p_x(t-) \dd N_x(\gamma t),
  \quad x\in \{1, \dots, n-1\},
  \end{aligned} \end{equation}
and at the boundaries by 
\begin{align}
     \dd   p_0(t) &=   \; \Big(q_1(t)-q_0(t) - \om_0^2 q_0 \Big) \dd   t
                     -
                    2  \gamma p_0(t) \dd t
                    +\sqrt{4  \gamma T_-} \dd \tilde w_-(t)
                    \vphantom{\Big(} \label{eq:pbdf} \\
  \dd   p_n(t) &=  \; \Big(q_{n-1}(t) -q_n(t) -\om_0^2 q_n(t) \Big)  \dd   t  +\cF_n(t)
                 \dd t - 2  p_n(t-) \dd N_n(\gamma t).
                     \vphantom{\Big(}  
\notag
\end{align}
Here $\Delta_N$ is the Neumann discrete laplacian, corresponding to the choice
$q_{n+1}:=q_n$ and {$q_{-1}= q_0$}.
We assume that the forcing $\cF_n(t)$ is $\theta_n$-periodic, with the
period $\theta_n=n^b\theta$, and the amplitude $n^{a}$, i.e.
\begin{equation}
\label{Fnt}
 \cF_n(t)= \;n^{a} \cF\left(\frac{t}{ \theta_n}\right).
\end{equation} Here $\theta>0$ and the scaling exponents 
 $a\in\bbR$, $b\ge0$ are to be adjusted later. We assume  $ \cF(t)$ is a
 smooth $1$-periodic function
such that
\begin{equation}
  \label{eq:2}
  \int_0^1  \cF(t) \dd t = 0, \qquad  \int_0^1  \cF(t)^2 \dd t > 0.
\end{equation}
Processes $\{N_x(t)\}$, $x=1,\ldots,n$
are independent, Poisson   of intensity $1$, while
$\tilde w_-(t)$ is a standard one dimensional Wiener process,
independent of the Poisson processes.
The parameter $\gamma>0$ 
regulates the intensity of the random perturbations
and the Langevin thermostat.
{We have choosen the same parameter in order to simplify notations,
  it does not affect the results {concerning} the macroscopic properties of the dynamics.}


The generator of the dynamics is given by
\begin{equation}
  \label{eq:7}
  \mathcal G_t =  \mathcal A_t +  \gamma S_{\text{flip}}
  + 2   \gamma S_-,
\end{equation}
where
\begin{equation}
  \label{eq:8}
  \mathcal A_t = \sum_{x=0}^n p_x \partial_{q_x}
  + \sum_{x=0}^n  (\Delta_N q_{x}-\om^2_0q_x) \partial_{p_x}
  +  \cF_n(t)  \partial_{p_n},
\end{equation}
and
\begin{equation}
  \label{eq:21}
   S_{\text{flip}} F (\pv,\qv) = \sum_{x=1}^{n}   \Big( F(\pv^x,\qv) - F(\pv,\qv)\Big),
 \end{equation}
 where $F:\bbR^{2(n+1)}\to\bbR$ is a bounded and measurable function,
 $\pv^x$ is the velocity configuration with sign flipped at the
 $x$ component, i.e. $\pv^x=(p_0',\ldots,p_n')$, with $p_y'=p_y$,
 $y\not =x$ and  $p_x'=-p_x$. Furthermore,
 \begin{equation}
   \label{eq:10}
   S_- = T_- \partial_{p_0}^2 - p_0 \partial_{p_0}.
 \end{equation}

The energy currents are given by
\begin{equation}
\label{eq:current}
 \mathcal G_t \mathcal E_x  = j_{x-1,x} - j_{x,x+1} ,
\end{equation} 
with 
$$
j_{x,x+1}:=- p_x (q_{x+1}- q_x) , \qquad \mbox{if }\quad x \in
\{0,...,n-1\}
$$  
and at the boundaries 
  \begin{equation} \label{eq:current-bound}
    { j_{-1,0} := 2 { \gamma} \left(T_- - p_0^2 \right)},
      \qquad
  j_{n,n+1} :=    -  \cF_n\left(t\right)   p_n.
\end{equation}

\subsection{Main results}
Our first result concerns the existence {and uniqueness}
of a periodic stationary state 
for the system. Fix $n\ge1$.
Following \cite{kasminski} Section 3.2,
we define
a \emph{periodic stationary probability measure}
{$\{\mu_t^P, t\in[0,+\infty)\}$ as
a solution of the forward equation $\partial_t \mu_t^P = \mathcal G_t^* \mu_t^P$
such that $\mu_{t + \theta_n}^P=\mu_{t}^P$. This condition is equivalent to }
\begin{equation}
  \label{eq:23}
  \int_0^{\theta_n} \dd s \int_{\bbR^{2(1+n)}} \mathcal G_s F(\rv, \pv) \mu_s^P(\dd\qv,\dd\pv)
    = 0,
\end{equation}
for any smooth test function {$F:\bbR^{2(n+1)}\to\bbR$}.

 Suppose that  $\{(\mathbf q(t), \mathbf p(t))\}_{t\ge0}$ is the
solution of \eqref{eq:flip}-\eqref{eq:pbdf}  initially distributed
according to $\mu_0^P$.
Given a measurable function $F:\bbR^{2(n+1)}\to\bbR$ integrable
w.r.t. each measure $\{\mu_s^P,
s\in[0,+\infty)\}$ we denote
\begin{equation}
\label{bar}
\bar F(t):=\bbE F\Big(\mathbf q(t), \mathbf
p(t)\Big)=\int_{\bbR^{2(n+1)}}F(\qv,\pv)\mu_t^P(\dd\qv,\dd\pv),\quad t\ge0,
\end{equation}
where $\bbE$ is the expectation w.r.t.   $\bbP$ - {the probability
measure corresponding to the noises and with initial data distributed by $\mu_0^P$.}
The function $\bar F(t)$ is $\theta_n$-periodic.
We denote its time average by
\begin{equation}
\label{ll}
\lang F\rang_n:= {\frac{1}{\theta_n}\int_0^{\theta_n}\bar F(t)\dd t}.
\end{equation}
The subscript $n$ in the notation of the average $\lang \cdot\rang_n$
will be sometimes omitted,
when it is obvious from the context.

\begin{theorem}
\label{periodic}
For a fixed $n\ge1$ there exists a unique periodic stationary state $\{\mu_s^P,
s\in[0,+\infty)\}$ for the system \eqref{eq:flip}-\eqref{eq:pbdf}.
 The measures $\mu_s^P$ are absolutely continuous with respect to the Lebesgue
 measure $\dd\qv\dd\pv$ and the density
 $\mu_s^P(\dd\qv,\dd\pv)=f_s^P(\qv,\pv) \dd\qv\dd\pv$ is
 strictly positive.
 {Furthermore $\min_x \lang p_x^2\rang_n \ge T_-$.}
\end{theorem}
The proof of the result is contained in Appendix \ref{appB}.

From \eqref{eq:current} we conclude that  the time averaged
energy current $\lang j_{x,x+1}\rang$ is constant for
$x=-1,\ldots,n$. Denote therefore
\begin{equation}  \label{eq:38}
  J_n^{a,b} :=  \lang j_{x,x+1}\rang,\quad x=-1,\ldots,n.
\end{equation}
In particular 
\begin{equation}  \label{eq:38a}
  J_n^{a,b}=
 -  \frac {n^{a}}{\theta_n}  \int_0^{\theta_n}  \cF\left(\frac{s}{\theta_n}\right)
 \bar p_n(s)\dd s  
  = {2} \gamma\Big(T_- - \lang p_0^2\rang \Big) .
\end{equation}
We prove, see Theorem \ref{thm-current} below, that if $b-a=1/2$, $a\le 0$
and $b\ge0$ (recall that $\theta_n=n^b\theta$), then 
\begin{equation}
  nJ_n^{a,b}=J^{a,b} + o(1), \quad\mbox{as $n\to+\infty$},
  \label{eq:76}
\end{equation}
 where $J^{a,b}<0$
is a  constant given by an explicit formula, see \eqref{051021-05}.

 In
our {first main result}, see Theorem \ref{th1} below, we prove the convergence
of the   {time averaged energy} functional  $\lang{\cal E}_x\rang$ to  a linear macroscopic profile
$$
      T(u) = T_--\frac{4\gamma Ju}{D},\quad\mbox{  $u\in[0,1]$},
$$ where
 the constant $D>0$ is given by
formula \eqref{eq:13}.

Our second result, see Theorem \ref{thm:var} below, deals with the
question of the vanishing of the fluctuations of the kinetic energy
functional {in the case when the period of the force is of a fixed
  microscopic size. More precisely, supposing that $b=0$ and $a=-1/2$, we
  prove that there exists a constant $C>0$ such that
\begin{equation}
    \label{eq:91cc}
  {    \sum_{x=0}^n \int_0^\theta \left(\bar{p_x^2}(t) - \lang
      p_x^2\rang \right)^2\dd t    
\le \frac{C}{n^2}, \quad n=1,2,\dots}
  \end{equation}}

\subsection{About the proof}
{
The macroscopic equation for the energy transport emerges from an
exact \emph{fluctuation-dissipation} decomposition of the energy current
\begin{equation}
  \label{eq:74}
  j_{x,x+1} = \mathcal G_t \frak f_x  - \frac 1{4\gamma} \nabla \mathfrak F_x,
\end{equation}
where $\mathfrak F_x$ and $\frak f_x$ are local second order polynomials
in the variables $\{q_{x+j}, p_{x+j}, j= -1, 0, 1\}$ (see \eqref{eq:32}, \eqref{eq:33}
and \eqref{072606-21a} for the precise definitions, valid also at the boundaries).
{After taking the time average we obtain}
\begin{equation}
\lang  j_{x,x+1} \rang
= - \frac 1{4\gamma} \nabla \lang \mathfrak F_x\rang.
\label{eq:75}
\end{equation}
Then we establish first \eqref{eq:76}, that takes care of the left hand side
of \eqref{eq:75}, i.e. $n\lang  j_{x,x+1} \rang \to J^{a,b}$.
The explicit calculation involved in computing $J^{a,b}$ {rely} only
on the
first moments of the periodic states.

Note that the equilibrium average at temperature $T$ of  $\mathfrak F_x$
equals $D T$ (with $D>0$  given by
  \eqref{eq:13}). Thus, all we need to prove is a kind of a local equilibrium
that allows to conclude that $\lang \mathfrak F_{[nu]}\rang \sim D T(u)$, $u\in [0,1]$.
This is the main part of the work. It involves {proving} the convergence
of the second moments of the positions and momenta. The most difficult
part is to establish an a priori  upper
bound for the time average of the total energy that proves it does not
grow faster than the size of the system $n$ (cf. \eqref{021911-21a}).

In \autoref{sec:cov} we {derive} a closed system of equations for the
time averages of the covariance matrix \eqref{163011-21} that involves
the discrete Neumann laplacian $\Delta_N$. After performing some
manipulations with these equation
we obtain that the time ageraged position covariances are given by the
Green's function $\left(\omega_0^2 - \Delta_N\right)^{-1}(x,x')$ plus
an error that is proportional to the averaged current, that is small
{(of order $O(1/n)$)}. In the bulk we have that
\begin{equation}
\left(\omega_0^2 - \Delta_N\right)^{-1}(x,x') \ \mathop{\longrightarrow}_{n\to\infty}\ 
\left(\omega_0^2 - \Delta\right)^{-1}(x-x').
\label{eq:77}
\end{equation}
Here $\Delta$ is the discrete laplacian on $\mathbb Z$,
that gives the covariance matrix of the positions of the system in equilibrium.  
But in order to obtain the correct bounds on the total energy we need
a careful control the
behavior of the Green's function at the boundaries, cf Lemma \ref{lem-ful} in 
\autoref{sec:green-fnc}.

Once the energy bound is established, the next step is to prove   local equilibrium,
which is contained in Proposition \ref{prop-loceq}. After this step
the proof of the main result follows directly, see 
\autoref{thm012912-21}.

In \autoref{appB} we prove the existence of the periodic measure.
The manipulations done with the equations for the covariance matrix use some ideas
from \cite{bll}.  The harmonic dynamics with self-consistent
Langevin reservoirs considered in that article has a similar covariance equations
of the corresponding stationary state, see \autoref{sec:conclusions}.

{Section \ref{sec:vanish-time-vari} contains the proof of the result
concerning the vanishing size of time variance of the kinetic energy,
see Theorem \ref{thm:var}.
\autoref{appC} is devoted to the presentation of the proofs of
auxiliary facts formulated in Section  \ref{sec:vanish-time-vari}.}}



{\subsection{Acknowledgements} We warmly thank David Huse for very stimulating
  discussions on the subject.  The work of J.L.L. was supported in
  part by the A.F.O.S.R. He thanks the Institute for Advanced Studies
  for its hospitality. T.K. acknowledges the support of the NCN grant 2020/37/B/ST1/00426. 
  S.O. has been partially supported by the ANR-15-CE40-0020-01
  grant LSD.}

\section{Some preliminaries and notation}






\subsection{The dynamics of  periodic means}

Define the averages in the periodic state:
\begin{equation}
  \label{eq:6}
  \begin{split}
  &\bar p_x(s)  :=  \int_{\bbR^{2(n+1)}}p_x \mu_s^P(\dd\qv,\dd\pv),\\
  &\bar q_{x}(s) := \int_{\bbR^{2(n+1)}}q_x \mu_s^P(\dd\qv,\dd\pv),\quad x=0,\ldots,n.
\end{split}
\end{equation}
They satisfy
\begin{equation} 
\label{eq:qdynamicsbulk-av}
\begin{aligned}
  \dot {  \bar q}_x &=  \bar p_x ,
  \\
  \dot{   \bar  p}_x &= \Big( \Delta_{\rm N}   -\om_0^2\Big)  \bar q_x-
  2{\gamma} \bar  p_x+\delta_{x,n}\cF_n\left(t\right)  
    , \; \quad x\in \{0, \dots, n\}.
  \end{aligned} \end{equation}
Here $\Delta_{\rm N}$ is the Neumann discrete laplacian, subject
to the boundary condition $q_{-1}=q_0$, $q_n=q_{n+1}$.
We can rewrite the above system using a matrix notation.
Let
$$
\bar{\bf q}(t)=
\left(
  \begin{array}{c}\bar q_0(t)\\
    \vdots
    \\
    \bar q_n(t)
    \end{array}\right),\qquad \bar{\bf p}(t)=
\left(
  \begin{array}{c}\bar p_0(t)\\
    \vdots
    \\
    \bar p_n(t)
    \end{array}\right).
$$
 and $({\bf q},{\bf p})$ be the vector of initial data.
We can write that
$$
\left(
  \begin{array}{c}\bar{\bf q}(t)\\
    \bar{\bf p}(t)
    \end{array}\right)=e^{-At}\left(
  \begin{array}{c}{\bf q}\\
    {\bf p}
    \end{array}\right)+  \int_0^t
e^{-A(t-s)}\;\cF_n\left(s\right){\rm e}_{p,n+1}\dd s.
$$
Here $A$ is a $2\times 2$ block matrix made of $(n+1)\times (n+1)$
matrices of the form
\begin{equation}
\label{A}
A=
\left(
  \begin{array}{cc}
    0&-{\rm Id}_{n+1}\\
    -\Delta_{\rm N} +\om_0^2& 2\ga {\rm Id}_{n+1}
  \end{array}
\right)
\end{equation}
where ${\rm Id}_{n+1}$ is  the $(n+1)\times (n+1)$ identity matrix.
We also let ${\rm e}_{q,\ell}$ and
${\rm e}_{p,\ell}$, $\ell=1,\ldots,n+1$ be the $1\times 2(n+1)$ column vectors  whose
components are given by  
\begin{equation}
\label{epq}
{\rm
  e}_{q,\ell,\ell'}=\delta_{\ell,\ell'}\quad\mbox{and}\quad  {\rm
  e}_{p,\ell,\ell'}=\delta_{n+1+\ell,\ell'},\quad \ell'=1,\ldots,2n+2.
\end{equation}
\begin{proposition}
\label{prop012212-21}
The
spectrum of matrix $A$ is contained in the half plane ${\rm Re}\,\la> 0$.
Thus, there exists $c>0$ such that
  \begin{equation}
\label{012212-21}
 (\la+A)^{-1}=\int_0^{+\infty} e^{-(\la+A)t}\dd t
  \end{equation}
  is well defined in the half-plane ${\rm Re}\,\la> -c$, with
  the integral on the right hand side of \eqref{012212-21} absolutely
  convergent. 
  \end{proposition}
  The proof of the result can be found in Appendix A of \cite{bll}.

\subsection{Time harmonics of the periodic means}
Consider the Fourier coefficients of the periodic means
\begin{equation}
  \label{eq:6a}
  \begin{split}
&\tilde p_x(\ell)=\frac{1}{\theta_n}\int_0^{\theta_n} e^{-2\pi i\ell
  t/\theta_n}\bar p_x(t)\dd t,\\
&\tilde q_x(\ell)=\frac{1}{\theta_n}\int_0^{\theta_n} e^{-2\pi i\ell
  t/\theta_n}\bar q_x(t)\dd t,\quad \ell\in\mathbb Z.
\end{split}
\end{equation}
They satisfy
\begin{equation} 
\label{eq:qdynamicsbulk-av-f}
\begin{aligned}
  \frac{2\pi i \ell}{\theta_n} \tilde q_x(\ell) &= \tilde p_x(\ell) ,
   \\
  \frac{2\pi i \ell}{\theta_n}  \tilde   p_x(\ell) &=  \Big(\Delta_{\rm N}   -\om_0^2\Big)  \tilde  q_x(\ell)-
  {2 \gamma} \tilde  p_x(\ell)+ {n^{a}} \tilde     \cF(\ell) \delta_{x,n}
    , \; \quad x\in \{0, \dots, n\}.
  \end{aligned} \end{equation}
Here
\begin{equation}
\label{cF}
\tilde     \cF(\ell)=\int_0^1 e^{-2\pi i\ell t }\cF(t)\dd t.
\end{equation}

In the particular case when $\ell=0$ we have 
\begin{equation}
  \label{eq:6aa}
  \lang q_x\rang_n  =  \tilde q_x(0)\quad\mbox{and}\quad  \lang p_x\rang_n
  =\tilde p_x(0)=0,\quad x=0,\ldots,n. 
\end{equation}
The last equality follows from the first equation of \eqref{eq:qdynamicsbulk-av}.
Combining the first and the second equations of
\eqref{eq:qdynamicsbulk-av-f}   we get
\begin{equation} 
\label{011005-21}
  0 = -L_{\om_0,\theta_n,\ell}^n \tilde q_x(\ell)
  +  \;  {n^{a}} \tilde     \cF(\ell) \delta_{x,n}
    , \; \quad x=0, \dots, n.
    \end{equation}
Here 
\begin{equation}
\label{Ln}
L_{\om_0,\theta,\ell}^n := 
\left[\omega_0^2 - \left(\frac{2\pi \ell}{\theta}\right)^2  + i\frac{4\pi
      \ell \gamma }{\theta}  \right]  -\Delta
\end{equation}
with the Neumann boundary conditions 
$
\tilde q_{-1}(\ell)=\tilde q_0(\ell)$, $\tilde q_n(\ell)=\tilde
q_{n+1}(\ell).
$

\subsection{Green's function corresponding to $L_{\om_0,\theta,\ell}^n$}

\label{sec2.6}
Denote by   $ G^n_{\om_0,\theta,\ell} (x,y)$ the Green's functions
corresponding to  $L_{\om_0,\theta,\ell}^n$. It is defined as the
solution of 
\begin{equation} 
\label{011205-21c}
 \  \delta_{x,y}=  L_{\om_0,\theta,\ell}^n  G^n_{\om_0, \theta,\ell}
    (x,y),\quad x,y=0,\ldots,n.
  \end{equation}
This function is given explicitly by
\begin{equation}
\label{Gpsi}
G^n_{\om_0, \theta,\ell} (x,y)=\sum_{j=0}^n\frac{\psi_j(x) \psi_j(y)}{\la_j+\om_0^2 -\left({2\pi \ell\theta^{-1}}\right)^2+{4\gamma \pi i \ell}{\theta^{-1}} }
\end{equation}
where  $\la_j$ and $\psi_j$ are the respective  eigenvalues
and eigenfunctions for the discrete Neumann laplacian $-\Delta_{\rm N}$. 
They are given by
\begin{align}
\label{laps}
\la_j=4\sin^2\left(\frac{\pi j}{2(n+1)}\right),\quad
\psi_j(x)=\left(\frac{2-\delta_{0,j}}{n+1}\right)^{1/2}\cos\left(\frac{\pi
    j(2x+1)}{2(n+1)}\right),\quad x,j=0,\ldots,n.
\end{align}
If $\ell=0$, then \eqref{Gpsi} defined the Green's function of
$\omega^2_0 -\Delta_{\rm N}$.  We will denote it by $G^n_{\om_0} (x,y)$.

\subsection{Green's function of the lattice laplacian}

\label{sec2.7}

Recall that the lattice gradient and laplacian of any $f:\bbZ\to\bbR$ are  defined as $\nabla
f_x=f_{x+1}- f_{x}$ and $\Delta f_x=f_{x+1}+
f_{x-1}-2 f_{x}$, $x\in\bbZ$, respectively.

Suppose that $\om_0>0$. Consider the
Green's function of $-\Delta + \omega^2_0$. It is given by, see
e.g. \cite[(27)]{ray},
\begin{align}
\label{GR}
&G_{\om_0}(x) = \left(-\Delta + \omega^2_0 \right)^{-1}(x)\quad 
=\int_0^1\left\{4\sin^2(\pi u)+\om_0^2\right\}^{-1}\cos(2\pi ux)du\\
&
 =\frac{1}{\om_0\sqrt{\om_0^2+4}}\left\{1+\frac{\om_0^2}{2}+\om_0\sqrt{1
    +\frac{\om_0^2}{4}}\right\}^{-|x|},\quad x\in\bbZ.\notag
\end{align}

\subsection{Some notation}

We adopt the following convention. For two sequences $(a_n)$ and
$(b_n)$ of real positive numbers we denote
$a_n\approx b_n$, $n\ge1$ if there exists $C>1$ such that $C^{-1}a_n\le
b_b\le Ca_n$, for all $n\ge1$.

\section{Periodic stationary energy transport}

\label{sec3}

\subsection{Asymptotics of the
time average of the mean current}

\label{ssec3.1}



Our first result gives an explicit formula for the asymptotics of the
time average of the mean current. In what follows we shall also be concerned with the functional
\begin{equation}
  \label{In}
I_n^{a,b}:=\frac{n^{a}}{\theta_n }\int_0^{\theta_n}\bar
q_n(t)\cF(t/\theta_n)\dd t,
\end{equation}
therefore we give its exact asymptotics, as $n\to+\infty$.
\begin{theorem}\label{thm-current}
Suppose that $\sum_{\ell}\ell^2|\tilde     \cF(\ell)|^2<+\infty$ and
\begin{equation}
\label{a-b}
b-a=\frac{1}{2},\quad a\le 0\quad\mbox{and}\quad b\ge0. 
\end{equation}
Then, 
\begin{equation}
\label{051021-05}
\lim_{n\to+\infty}n J_n^{a,b}=  J^{a,b}
:=-\left(\frac {2\pi}{\theta}\right)^2   \sum_{\ell\in\bbZ} \ell^2{\cal Q}^{a,b}(\ell), 
\end{equation}
with ${\cal Q}^{a,b}(\ell)$ given  by, cf \eqref{cF}, 
\begin{equation} 
\label{021205-21f1}
\begin{aligned}
  & {\cal Q}^{-1/2,0}(\ell)= 4\gamma|\tilde     \cF(\ell)|^2
  \int_0^1 \cos^2\left(\frac{\pi z}{2}\right)
  \left\{\left[4\sin^2\left(\frac{\pi z}{2}\right)
      +\om_0^2 -\left(\frac{2\pi\ell}{\theta}\right)^2\right]^2
    +\left(\frac{{4} \gamma \pi \ell}{\theta}\right)^2 \right\}^{-1}\dd z
\end{aligned} \end{equation}
and 
\begin{equation} 
\label{021205-21f2}
\begin{aligned}
   {\cal Q}^{b-1/2,b}(\ell)= 4\gamma|\tilde     \cF(\ell)|^2
   \int_0^1 \cos^2\left(\frac{\pi z}{2}\right)
   \left[4\sin^2\left(\frac{\pi z}{2}\right)
       +\om_0^2 
     \right]^{-2}\dd z
,\quad\mbox{when }b>0.
\end{aligned} 
\end{equation}
Furthermore, we have
\begin{equation}
    \label{022801-22}
I_n^{a,a-1/2}={\frak I}^{a,a-1/2}{n^{2a}}+o(n^{2a}),
  \end{equation}
  where
  \begin{align*}
  &{\frak I}^{0,-1/2}:=  2\sum_\ell  |\tilde \cF(\ell)|^2
    \int_0^1  \cos^2\left(\frac{\pi z
    }{2}\right) \left\{ 4\sin^2\left(\pi
   z\right)+\left[ \om_0^2 -\left(\frac{2\pi \ell}{\theta}\right)^2 \right]\right\}
  \\
  &\times \left\{\left[4\sin^2\left(\frac{\pi z}{2}\right)+\om_0^2 -\left(\frac{2\pi \ell}{\theta}\right)^2\right]^2 +\left(\frac{{4}\gamma \pi  \ell}{\theta} \right)^2 
  \right\}^{-1}dz
  \end{align*}
  and for $b>0$
  \begin{align*}
  {\frak I}^{b-1/2,b}:=  2 \sum_\ell  |\tilde \cF(\ell)|^2
    \int_0^1  \cos^2\left(\frac{\pi z
    }{2}\right) \left\{ 4\sin^2\left(\pi
   z\right)+ \om_0^2 \right\}
  \left\{4\sin^2\left(\frac{\pi z}{2}\right)+\om_0^2
  \right\}^{-2}dz.
  \end{align*}
\end{theorem}
The proof of the results is given in Section
\ref{sec:proof-theorem-refthm}.

\begin{remark}
  Calculations made in Appendix \ref{appD} show that
  \begin{align*}
     &{\cal Q}^{-1/2,0}(\ell) =\frac{\theta |\tilde     \cF(\ell)|^2}{2\pi\ell}
    {\rm Im}\left( \left\{\frac{2 }{\la(\om_0)\sqrt{1+4/\la(\om_0)}}+\frac12\right\} \left\{1+\frac{ \la(\om_0)}{2}\Big(1+\sqrt{1
     +\frac{4}{\la(\om_0)}}\Big)\right\}^{-1}\right),
\end{align*}
with
$$
\la(\om_0):=\om_0^2 -\left(\frac{2\pi\ell}{\theta}\right)^2 
    +i\left(\frac{{4} \gamma \pi \ell}{\theta}\right) .
    $$
    Furthermore,
    \begin{align*}
    {\cal Q}^{b-1/2,b}(\ell)
    = 
    \frac{2\gamma|\tilde     \cF(\ell)|^2  (4 +\om_0^2) }
    { (\om_0^4+4 \om_0^2+8)^{3/2}} .
\end{align*}
  \end{remark}

From Theorem \ref{thm-current} and the definition of $j_{-1,0}$, see
\eqref{eq:current-bound}, we immediately conclude the following.
\begin{corollary}
\label{thm021911-21}
We have
\begin{equation}
\label{031911-21}
 T_- -\lang p_0^2\rang_n=\frac{J^{a,b}}{ 2\ga n}+o\left(\frac1n\right)
\end{equation}
\end{corollary}


\subsection{Asymptotic profile of the periodic averages of the means
  of the energy function}
\label{sec:non-stat-diff}

 The following  result holds.
\begin{theorem}\label{th1}
Under the assumptions of Theorem \ref{thm-current}
  we have {
  \begin{equation}
    \label{eq:3}
    \lim_{n\to\infty} \frac 1n \sum_x \varphi\left(\frac xn \right) \lang p^2_x\rang
    {=  \lim_{n\to\infty} \frac 1n \sum_x \varphi\left(\frac xn \right)
    \lang {\cal E}_x\rang}
    = \int_0^1 \varphi (u) T(u) \dd u,
  \end{equation}
 } with 
  \begin{equation}
    \label{eq:5}
      T(u) = T_--\frac{4\gamma Ju}{D}, \quad u\in[0,1],
  \end{equation}
for any $\varphi\in C[0,1]$.
  Here $J$ is given by \eqref{051021-05}
and    $D$  is defined by, cf \eqref{GR},
  \begin{equation}
    D = 1 - \omega_0^2 \Big(G_{\om_0}(0)+ G_{\om_0}(1)\Big).
    \label{eq:13}
  \end{equation}
\end{theorem}
We present the proof of the theorem in Section \ref{sec-converg}.

A simple calculation, using  \eqref{GR}, yields an explicit
formula for the coefficient $D$, cf \cite[(4.18)]{bll},
\begin{equation}
D =\frac{2}{2+\om_0^2+\om_0\sqrt{\om_0^2+4}}.
  \label{D}
  \end{equation}
Therefore,
$$
\lim_{\om_0\to+\infty}D(\om_0)=0, \qquad \lim_{\om_0\to 0+}D(\om_0)=1 .
$$
{
  In the case $a= -1/2, b=0$ and $\theta$ constant, we can improve the statement
  of the theorem \eqref{th1} as we do not neet the time average over the period,
  more precisely we prove in \autoref{sec:vanish-time-vari} that
  \begin{equation}
    \label{eq:84}
    \lim_{n\to\infty} \frac 1n \sum_x \varphi\left(\frac xn \right) \int_0^\theta
    \left( \bar{p^2_x}(t) - T\left(\frac xn\right)\right)^2 dt
    = 0.
  \end{equation}
  }







\section{The first moments of positions and momenta}

\subsection{Proof of Theorem \ref{thm-current}}
\label{sec:proof-theorem-refthm}

\bigskip

We only show \eqref{051021-05}. The proof of \eqref{022801-22} follows
from analogous calculations. Using \eqref{eq:38a}, the Plancherel identity and then
the first equation of \eqref{eq:qdynamicsbulk-av-f} we get
\begin{equation}
  \label{eq:14}
  \begin{split}
      J_n^{a,b} &= 
 -  n^{a} \sum_\ell \tilde \cF(\ell) \tilde p_n(\ell)^*
   \\
& = \frac{ 2\pi i n^{a}}{\theta_n} \sum_\ell \tilde
    \cF(\ell) \ell \tilde q_n(\ell)^* =   \frac{ 2\pi i n^{a}}{\theta_n}\sum_\ell \sum_x \delta_{x,n}
    \tilde \cF(\ell) \ell \tilde q_x(\ell)^*.
  \end{split}
\end{equation}
Thanks to
 \eqref{011005-21}  we can  further write  
\begin{equation}
  \label{eq:14a}
  \begin{split}
    &  J_n^{a,b}  = \frac{2\pi i}{\theta_n} \sum_\ell \sum_x \ell \tilde q_x(\ell)^* L_{\om_0,\theta_n,\ell}^n\tilde q_x(\ell)\\
    &= \frac{2\pi i}{\theta_n}\sum_\ell \ell\sum_x
    \left\{|\nabla \tilde q_x(\ell)|^2 +\left[ \om_0^2 -\left(\frac{2\pi \ell}{\theta_n}\right)^2 \right]
      |\tilde q_x(\ell)|^2 \right\}
    - 2 \gamma\left(\frac{2\pi }{\theta_n}\right)^2 \sum_\ell \ell^2 \sum_x |\tilde q_x(\ell)|^2.
  \end{split}
\end{equation}
Using the parity of $|\nabla \tilde q_x(\ell)|^2  +
\left[\om_0^2 -\left(\frac{2\pi \ell}{\theta_n}\right)^2 \right]
|\tilde q_x(\ell)|^2$
we get
\begin{equation}
  \label{eq:16}
    J_n^{a,b} = - 2\ga\left(\frac{2\pi }{\theta_n}\right)^2 \sum_\ell
    \ell^2 \sum_x   |\tilde q_x(\ell)|^2.
\end{equation}


From \eqref{011005-21} we have
\begin{equation} 
\label{021205-21}
  \; n^a \tilde     \cF(\ell) \delta_{x,n}
  =  L^n_{\om_0,\theta_n,\ell}\tilde
    q_x(\ell),\quad x=0,\ldots,n .
  \end{equation}
  
Hence, by  \eqref{011205-21c},
\begin{equation} 
\label{021205-21a}
    \; n^a\tilde     \cF(\ell) G^n_{\om_0, \theta_n,\ell} (x,n)
     =  \tilde
    q_x(\ell),\quad x=0,\ldots,n .
  \end{equation}
  
It follows, by \eqref{Gpsi},
\begin{equation}
\label{012202-22}
  \begin{split}
  &  \sum_{x=0}^n |\tilde q_x (\ell)|^2 =
     {n}^{2a} |\tilde \cF(\ell)|^2 \sum_{x=0}^n [G^n_{\om_0,\theta_n,\ell} (x,n)]^2
= \sum_{j=0}^n\frac{ {n}^{2a} |\tilde \cF(\ell)|^2\psi_j^2(n)   }{[\la_j+\om_0^2 -\left({2\pi \ell}{\theta_n^{-1}}\right)^2]^2+\left({4\gamma \pi  \ell}{\theta_n^{-1}} \right)^2 }
     .
  \end{split}
\end{equation}
A straightforward calculation, using formula \eqref{Gpsi}, yields
\begin{align*}
G^n_{\om_0,\theta_n,\ell}(0,n)&=
\frac{2}{n+1} \sum_{l=1}^{n} (-1)^k\cos^2\left(\frac{\pi
    k}{2(n+1)}\right) \left\{4\sin^2\left(\frac{\pi
  k}{2(n+1)}\right)+\om_0^2 -\left(\frac{2\pi
  \ell}{\theta_n}\right)^2+\frac{4\gamma \pi i
  \ell}{\theta_n}\right\}^{-1}\\
&
+ O\left(\frac{1}{n+1}\right) =o(1).
\end{align*}
Therefore, we conclude that
\begin{equation}
\label{010610-21}
\begin{split}
&\sum_x  |\tilde q_x(\ell)|^2=
2|\tilde     \cF(\ell)|^2 {n}^{2a} \int_0^1 \cos^2\left(\frac{\pi z
    }{2}\right) \\
&
\times
\left\{\left[4\sin^2\left(\frac{\pi z}{2}\right)+\om_0^2 -\left(\frac{2\pi \ell}{\theta_n}\right)^2\right]^2 +\left(\frac{4\gamma \pi  \ell}{\theta_n} \right)^2 
  \right\}^{-1}\dd z+o\left({n}^{2a} \right).
\end{split}
\end{equation}
From \eqref{eq:16} and \eqref{010610-21}   we get
that for $a,b$ satisfying \eqref{a-b}   
\begin{equation} 
\label{021205-21f}
\begin{aligned}
&  J_n^{a,b} 
=-{\frac{4\gamma}{n}\left(\frac{2\pi }{\theta}\right)^2\sum_{\ell}\ell^2}|\tilde     \cF(\ell)|^2 \int_0^1 \cos^2\left(\frac{\pi z
    }{2}\right) \\
&
\times\left\{\left[4\sin^2\left(\frac{\pi z}{2}\right)+\om_0^2 -\left(\frac{2\pi \ell}{\theta_n}\right)^2\right]^2 +\left(\frac{{4}\gamma \pi  \ell}{\theta_n} \right)^2 
  \right\}^{-1}\dd z+o\left(\frac{1}{n}\right)\notag
\end{aligned} \end{equation}
and Theorem \ref{thm-current} follows.\qed

\subsection{$L^2$ norms of the position and momentum averages}

Denote
\begin{equation}
\label{l2}
\begin{split}
&{\lang \bar q_x^2\rang}:=\frac{1}{\theta_n}\int_0^{\theta_n}  \bar q_x^2(s) \dd
s=\sum_{\ell}  |\tilde q_x(\ell)|^2,\\
&
{\lang\bar p_x^2\rang} :=\frac{1}{\theta_n}\int_0^{\theta_n}  \bar p_x^2(s) \dd
s=\sum_{\ell}  |\tilde p_x(\ell)|^2 .
\end{split}
\end{equation}
Using \eqref{eq:qdynamicsbulk-av-f} we get
\begin{equation}
\label{011012-21}
\sum_x  |\tilde
p_x(\ell)|^2=\left(\frac{2\pi\ell}{\theta_n}\right)^2\sum_x  |\tilde q_x(\ell)|^2.
\end{equation}
By virtue of \eqref{010610-21} we get
\begin{equation}
 \label{011012-21ss}
\begin{split}
 &\sum_x   |\tilde p_x(\ell)|^2=\frac{2 }{n}
\left(\frac{2\pi\ell |\tilde \cF(\ell)|}{\theta}\right)^2
 \int_0^1 \cos^2\left(\frac{\pi z
    }{2}\right) \\
&
\times
\left\{\left[4\sin^2\left(\frac{\pi z}{2}\right)+\om_0^2 -\left(\frac{2\pi \ell}{\theta_n}\right)^2\right]^2 +\left(\frac{4\gamma \pi  \ell}{\theta_n} \right)^2 
  \right\}^{-1}\dd z+o\left(\frac1n\right).
 \end{split}\end{equation}
We have shown therefore the following.
\begin{proposition}
\label{prop011012-21}
Under the assumptions of Theorem \ref{thm-current} we have
\begin{equation}
\label{021012-21}
\sum_{x=0}^n\lang \bar q_x^2\rang\approx
n^{2a}\quad\mbox{and}\quad 
\sum_{x=0}^n\lang \bar p_x^2\rang\approx \frac{1}{n},\quad n\ge1.
\end{equation}
\end{proposition}

\section{The second moments for the momentum and
  position variables}
\label{sec:fluct-diss-relat}

\subsection{Fluctuation-dissipation relations}

Define
\begin{equation}
  \label{eq:32}
  \begin{split}
  \frak f_x &:= \frac 1{4\gamma} \left(q_{x+1} - q_x\right)\left(p_x + p_{x+1}\right)
  + \frac 14 \left(q_{x+1} - q_x\right)^2 ,\qquad x=0, \dots, n-1,\\
\mathfrak F_x &= p_x^2 + \left(q_{x+1} - q_x\right) \left(q_{x} - q_{x-1}\right)
    -\omega_0^2 q_x^2,\qquad x=0,\dots, n,
\end{split}
\end{equation}
with the convention that $q_{-1} = q_{0}$, $q_n=q_{n+1}$.
Then
\begin{equation}
  \label{eq:33}
    \mathcal G_t \frak f_x = \frac 1{4\gamma} \nabla
    \mathfrak F_x +j_{x,x+1}  +\frac {\delta_{x,n-1}} {4\gamma} n^{a}\cF(t/\theta_n) \left(q_{n} - q_{n-1}\right),\quad x=0, \dots, n-1.
\end{equation}

After performing the expectation and time averaging  we get
\begin{align}
  \label{072606-21a}
  &\lang \frak F_x\rang=\lang \frak F_0\rang+
  \sum_{y=0}^{x-1}\lang \nabla \frak F_y\rang =\lang \frak F_0\rang -4\gamma  J_n{ x}
\\
&
+\frac{\delta_{x,n}
  n^{a}}{\theta_n}\int_0^{\theta_{n}}\cF(t/\theta_n) \left(
  \bar q_{n-1}(t)-\bar q_{n}(t) \right)\dd t, \qquad x=1,\ldots,n.\notag
\end{align}
\begin{remark}
{\em Note that the expectation of $\frak F_x$ with respect to  the Gibbs
Gaussian measure on the lattice $\bbZ$, with the Hamiltonian
$\sum_{x\in\bbZ} {\cal E}_x$ and  the inverse temperature $T^{-1}$, is given by
\begin{align}
\label{010803-22}
 T\Big[ 1 - G_{\om_0}(1) - G_{\om_0}(0) + 2 G_{\om_0}(1) - \omega_0^2
  G_{\om_0}(0) \Big].
\end{align}
Since $G_{\om_0}$ is the Green's function for $\om_0^2-\Delta$, where
$\Delta$ is the free lattice laplacian, we have 
\begin{equation}
\label{D1}
 1 - G_{\om_0}(2) - G_{\om_0}(0) + 2 G_{\om_0}(1) - \omega_0^2
  G_{\om_0}(0)= 1- \omega_0^2\Big(G_{\om_0}(0)+G_{\om_0}(1)\Big)= D
\end{equation}
and the expression in \eqref{010803-22} equals $DT$.}
\end{remark}

\section{The covariance matrix of the periodic state}
\label{sec:cov}

\subsection{Dynamics of fluctuations}

Denote
\begin{equation}
\label{pqpp}
q'_x(t):=q_x(t)-\bar q_x(t)\quad\mbox{and}\quad  p'_x(t):=p_x(t)-\bar p_x(t)
\end{equation}
for $x=0,\ldots,n$.
From \eqref{eq:flip}-\eqref{eq:pbdf} and
\eqref{eq:qdynamicsbulk-av} 
we get 
\begin{equation} 
\label{eq:fflip-1}
\begin{aligned}
  \dot { q}'_x(t) &= p'_x(t) ,
  \qquad \qquad \qquad  \qquad \qquad x\in \{0, \dots, n\},\\
  \dd { p}'_x(t) &=  \left(\Delta q'_x-\om_0^2 
    q_x'\right) \dd t-   2\ga  {p_x'(t)} \dd t-   2 p_x(t-) \dd\tilde N_x(\gamma t),
  \quad x\in \{1, \dots, n\},
  \end{aligned} \end{equation}
and at the left boundary
\begin{align}
     \dd  p'_0(t) =   \left(\Delta q'_0-\om_0^2 
    q_0'\right)  \dd   t
  -  2  \gamma  p'_0(t) \dd t
                    +\sqrt{4 \gamma T_-} \dd \tilde w_-(t).
                    \vphantom{\Big(}   \label{eq:fpbdf-2}
\end{align}
Here $
\tilde N_x(t):=N_x(t)-t.
$
Let 
$$
{\bf X}(t)=\left(\begin{array}{c}
{\bf q}(t)\\
{\bf p}(t)
\end{array}\right),\quad \bar {\bf X}(t)=\left(\begin{array}{c}
\bar{\bf q}(t)\\
\bar{\bf p}(t)
\end{array}\right),
$$
and $
{\bf X}'(t)=  {\bf X}(t)-\bar {\bf X}(t).
$ 
Furthermore, we define
\begin{equation}
  \begin{split}
\Sigma ({\bf p}) =\left[\begin{array}{cc}
0_{n+1}&0_{n+1}\\
0_{n+1}&D({\bf p})
                        \end{array}\right],
                      \,\mbox{with }
D({\bf p}) =
       \begin{bmatrix}
\sqrt{4 \gamma T_-} & 0 & 0 &\dots&0\\
                     0& -2  p_1 &  0 &\dots&0\\
                     0 & 0 & -2  p_2 &\dots&0\\
                     \vdots & \vdots & \vdots & \vdots&\vdots\\
 0& 0 & 0 & \dots &-2  p_n
                        \end{bmatrix}.
 \end{split}
\label{eq:22}
\end{equation}
The symbol $0_{n}$ denotes the
null $n\times n$ matrix.
The solution of  \eqref{eq:fflip-1}--\eqref{eq:fpbdf-2} satisfies
\begin{equation}
  \label{Xts0}
{\bf X}'(t)=e^{-At}{\bf X}'(0)+\int_{0}^t
e^{-A(t-s)}\Sigma \Big({\bf p}(s-)\Big)\dd M(s),\quad t\ge0.
\end{equation}
Here $A$ is defined by \eqref{A} and
$\left(M(t)\right)_{t\ge0}$ is $2(n+1)$-dimensional vector martingale 
{
$$
\dd M(s) =\left(\begin{array}{c}
                  0\\
                  \vdots\\
                  0\\
                  {\dd \tilde w(s)}\\
                  \dd \tilde N_1(\gamma s)\\
                  \vdots\\
                  \dd \tilde N_n(\gamma s)
\end{array}\right).
$$
}


Suppose that ${\bf X}$ is a random vector that is independent of the noise
 and distributed according to $\mu_0^P$.  Denote by $\bar {\bf X}$ the
 vector of its means and by ${\bf X}'={\bf X}-\bar {\bf X}$.
For any $\ell\ge0$ define ${\bf X}_\ell(t)$, $t\ge -\ell\theta$ - the solution of
\eqref{eq:fflip-1}--\eqref{eq:fpbdf-2}
that satisfies ${\bf X}_\ell(-\ell \theta)={\bf X}$. 
We call such solutions {\em $\theta$-periodic}.
 Note that
\begin{equation}
  \label{Xts}
{\bf X}'_\ell(t)=e^{-A(t+\ell\theta)}{\bf X}'+\int_{-\ell\theta}^t
e^{-A(t-s)}\Sigma \Big({\bf p}_\ell(s-)\Big)\dd M(s),\quad t\ge-\ell\theta.
\end{equation}
Obviously,
$\Big({\bf X}_\ell(t+\theta)\Big)_{t\ge-\ell\theta}$ has the same law
as  $\Big({\bf X}_\ell(t)\Big)_{t\ge-\ell\theta}$.

\subsection{The covariance matrix}

Suppose that ${\bf X}(t)$ is a $\theta$-periodic solution of
\eqref{eq:flip}-\eqref{eq:pbdf}.
Define the vector valued function
$$
{\bar{{\frak p}^2}(t)} =\left(
  \begin{array}{c}
    \bbE p_0^2(t)\\
    \vdots\\
    \bbE p_n^2(t)
  \end{array}
    \right) 
    $$
    and the covariance matrix
\begin{equation}
\label{S1ts}
S(t)
=\left[
  \begin{array}{cc}
    {S^{(q)}(t)}&S^{(q,p)}(t)\\
   S^{(p,q)}(t)& S^{(p)}(t)
  \end{array}
\right],
\end{equation}
where
\begin{align}
\label{S1ts1}
&S^{(q)}(t)=\Big[\bbE[q_x'(t)q_y'(t)]\Big]_{x,y=0,\ldots,n},\quad S^{(q,p)}(t)=\Big[\bbE[q_x'(t)p_y'(t)]\Big]_{x,y=0,\ldots,n},\notag\\
&\\
&
S^{(p)}(t)=\Big[\bbE[p_x'(t)p_y'(t)]\Big]_{x,y=0,\ldots,n}\quad \mbox{and}\quad S^{(p,q)}(t)=\Big[S^{(q,p)}(t)\Big]^T. \notag
\end{align}
Obviously both $\bar{{\frak p}^2}(t)$,  $S^{(q)}(t)$,
$S^{(p)}(t)$ are $\theta$-periodic  and the matrices are symmetric.

We shall also consider $\lang  \bar{{\frak p}^2}\rang$ and
$\lang S^{(\alpha)}\rang$, $\al\in\{q,p, (p,q), (q,p) \}$ the respecive vector and matrices of time
averages.

 Given a vector $\frak y = (y_0, y_1, \dots, y_n)$, define also the matrix
 valued function
 \begin{equation}
\label{D2}
D_2({\frak y}) =
       4 \gamma\begin{bmatrix}
  T_-& 0 & 0 &\dots&0\\
                     0&   y_1 &  0 &\dots&0\\
                     0 & 0 &   y_2 &\dots&0\\
                     \vdots & \vdots & \vdots & \vdots&\vdots\\
 0& 0 & 0 & \dots &  y_n
                        \end{bmatrix}.
\end{equation}
Let $\Sigma_2 ({\frak y}) $ be the $2\times 2$ block matrix
 \begin{equation}
\label{S2}
\Sigma_2 ({\frak y}) =\left[\begin{array}{cc}
0_{n+1}&0_{n+1}\\
0_{n+1}&D_2({\frak y})
\end{array}\right].
\end{equation}
    \begin{proposition}
      \label{prop010612-21}
      The following identities hold
      \begin{equation}
        \label{010612-21}
       { S(t)=\int_{0}^{+\infty}e^{-As} \Sigma_2
         \big(\overline{\frak p^2}(t-s)\big) e^{-A^Ts}\dd s},
       \quad t\ge0
      \end{equation}
      and
      \begin{equation}
        \label{010612-21aa}
        \lang S\rang
        =\int_{0}^{+\infty}e^{-As}\Sigma_2\big(\lang {\frak
          p}^2\rang\big)e^{-A^Ts}\dd s.
      \end{equation}
      \end{proposition}
      \proof
      Formula \eqref{010612-21aa} is an obvious consequence of
      \eqref{010612-21}, so we only prove the latter. From \eqref{Xts}
      we conclude that
      \begin{equation}
  \label{Xtsl}
  \begin{split}
  &  S(t)=e^{-A(t+\ell\theta)}\bbE\Big\{{\bf X}'(0)\otimes \Big({\bf
    X}'(0)\Big)^T \Big\}e^{-A^T(t+\ell\theta)}\\
  &
  +\int_{-\ell\theta}^t
  e^{-A(t-s)} \Sigma_2 \Big(\overline{\frak {\bf p}^2}(s)\Big) e^{-A^T(t-s)}\dd s,
  \quad t\ge-\ell\theta.
\end{split}
\end{equation}
Letting $\ell\to+\infty$ and using  Proposition \ref{prop012212-21} we obtain
  \begin{equation}
  \label{Xtsl1}
    S(t)=\int_{-\infty}^t
    e^{-A(t-s)} \Sigma_2 \Big(\overline{\frak {\bf p}^2}(s)\Big)
    e^{-A^T(t-s)}\dd
s,\quad  t\in\bbR.
\end{equation}
Changing variables $s':=t-s$ we conclude \eqref{010612-21}.
      \qed

      \subsection{Structure of the covariance matrix
        of the periodic averages}

     { By \eqref{010612-21aa} and partial integration in time, we have
      \begin{equation}\label{eq:covev}
        \begin{split}
          &AS(t) = -\int_0^\infty \left(\frac{d}{ds}  e^{- As} \right) \Sigma_2
          \big(\overline{\frak p^2}(t-s)\big) e^{-A^Ts}\dd s\\
          =& -\int_0^\infty  e^{- As} \Sigma_2 \big(\overline{\frak p^2}(t-s)\big)
        A^T e^{-A^Ts}\dd s +
        \Sigma_2\big(\overline{\frak p^2}(t)\big) 
        -\int_0^\infty  e^{- As} \Sigma_2 \big(\overline{\frak p^2}' (t-s)\big) e^{-A^Ts} \dd s\\
        &= -S(t) A^T +  \Sigma_2\big(\overline{\frak p^2}(t)\big) - S'(t).
        \end{split}
      \end{equation}
Integrating the above relation over the period we 
conclude that the  
matrix $\lang S\rang$ satisfies the equation}
\begin{equation}
  \label{SA}
A \lang S\rang+ \lang S\rang A^T=\Sigma_2(\lang {\frak p^2}\rang).
\end{equation}
It leads to the following equations on the blocks defined in \eqref{S1ts}
(see \eqref{A} and \eqref{D2}):
\begin{align*}
  &\lang S^{(q,p)}\rang=\Big[\lang S^{(p,q)}\rang\Big]^T
    =-\lang S^{(p,q)}\rang,\\
  &\lang S^{(q)}\rang (\om_0^2-\Delta_{\rm N})+ 2\gamma \lang S^{(q,p)}\rang
    - \lang S^{(p)}\rang=0,\\
  &(\om_0^2-\Delta_{\rm N}) \lang S^{(q)}\rang+ 2\gamma \lang S^{(p,q)}\rang
    - \lang S^{(p)}\rang=0\\
  &(\om_0^2-\Delta_{\rm N}) \lang S^{(q,p)}\rang-\lang S^{(q,p)}\rang
    (\om_0^2-\Delta_{\rm N})=D_2(\lang {\frak p^2}\rang) - 4\gamma \lang S^{(p)}\rang.
\end{align*}
From here we conclude
\begin{align}
  \label{163011-21}
  &\lang S^{(q,p)}\rang=-\lang S^{(p,q)}\rang,\notag\\
 &\lang S^{(p)}\rang =\frac12\big\{\lang S^{(q)}\rang  (\om_0^2-\Delta_{\rm N})
   +(\om_0^2-\Delta_{\rm N}) \lang S^{(q)}\rang \Big\}
   , \notag\\
  &4\gamma \lang S^{(q,p)}\rang
    =(\om_0^2-\Delta_{\rm N}) \lang S^{(q)}\rang-
   \lang S^{(q)}\rang (\om_0^2-\Delta_{\rm N})  ,\\
  &(\om_0^2-\Delta_{\rm N}) \lang S^{(q,p)}\rang-\lang S^{(q,p)}\rang (\om_0^2-\Delta_{\rm N})
    =D_2(\lang {\frak p^2}\rang) - 4\gamma \lang S^{(p)}\rang
   .\notag
\end{align}

Denote
$$
 \tilde{ S}^{(q,p)}_{j,j'}=\sum_{x,x'=0}^n \lang S^{(q,p)}_{x,x'}\rang\psi_j(x) \psi_{j'}(x')
$$
and analogously define $ \tilde{ S}^{(p)}_{j,j'}$ and $ \tilde{
  S}^{(q)}_{j,j'}$. The eigenvalues of $\om_0^2-\Delta_{\rm N}$ are
given by
{
\begin{equation}
\label{muj}
\mu_j=\om_0^2+\la_j=\om_0^2+4\sin^2\left(\frac{\pi
    j}{2(n+1)}\right),\quad j=0,\ldots,n.
\end{equation}}
Then we have the inverse relations
\begin{equation}
  \label{eq:28}
  \lang S^{(\alpha)}_{x,x'}\rang =\sum_{j,j'=0}^n \tilde{ S}^{(\alpha)}_{j,j'}
  \psi_j(x) \psi_{j'}(x').
\end{equation}

With this notation
we can rewrite \eqref{163011-21} as follows
\begin{align}
  \label{163011-21a}
  &\tilde S^{(q,p)}_{j,j'}
   =-\tilde S^{(p,q)}_{j,j'},\notag\\
 &\tilde S^{(p)}_{j,j'}=\frac12\big(\mu_j+\mu_{j'}\big)\tilde
   S^{(q)}_{j,j'}
   \notag\\
  &4\ga \tilde S^{(q,p)}_{j,j'} 
    = \tilde S^{(q)}_{j,j'}(\mu_j-\mu_{j'}),\\
  &(\mu_j-\mu_{j'}) \tilde S^{(q,p)}_{j,j'}=4 \ga \tilde F_{j,j'}
    -4\ga\tilde  S^{(p)}_{j,j'}
    . \notag
\end{align}
where
            \begin{equation}
                 \tilde F_{j,j'}:=
                 \sum_{y=0}^n\psi_j(y)\psi_{j'}(y)\lang p_y^2\rang
                  +\Big( T_--  \lang p_0^2\rang \Big)\psi_j(0)\psi_{j'}(0).\label{eq:49}
                \end{equation}

By eliminating $\tilde S^{(q,p)}_{j,j'} $ from the above equations we get
\begin{equation}
  \label{eq:26}
  \tilde S^{(p)}_{j,j'}=  \tilde F_{j,j'} - \frac{(\mu_{j} - \mu_{j'})^2}{(4\gamma)^2}
   \tilde S^{(q)}_{j,j'}.
\end{equation}
Thus,
\begin{equation}
  \label{eq:27}
  \tilde S^{(q)}_{j,j'} = \frac{2}{\mu_{j} + \mu_{j'}} \tilde F_{j,j'}
  -  \frac{(\mu_{j} - \mu_{j'})^2}{8\gamma^2(\mu_{j} + \mu_{j'})}
   \tilde S^{(q)}_{j,j'}.
\end{equation}
It follows that
\begin{equation}
  \label{eq:47}
  \tilde S^{(p)}_{j,j'} = \Theta(\mu_j,\mu_{j'}) \tilde F_{j,j'} , \qquad
  \Theta(\mu_j,\mu_{j'}) =
  \left[1 + \frac{(\mu_j -\mu_{j'})^2}{8\gamma^2 (\mu_j + \mu_{j'})}\right]^{-1},
\end{equation}
and
\begin{equation}
  \label{eq:69}
  \tilde S^{(q)}_{j,j'} = \frac{ 2  \Theta(\mu_j,\mu_{j'})}{\mu_{j} + \mu_{j'}} \tilde F_{j,j'} .
\end{equation}

 \section{Energy Bounds}
\label{sec:bound-kinet-energ}

{Throughout the remainder of the paper we shall always assume that
  the assumptions of Theorem \ref{thm-current} are in force}. From \eqref{eq:47} and \eqref{eq:49} we have
\begin{equation}
  \label{eq:50}
  \begin{split}
  \lang S^{(p)}_{x,x} \rang = \sum_{j,j'}  \Theta(\mu_j,\mu_{j'}) \tilde F_{j,j'}
  \psi_j(x) \psi_{j'}(x) 
  = \sum_y M_{x,y}  \lang p_y^2 \rang
  + \big(T_--\lang p_0^2\rang \big) M_{x,0}
\end{split}
\end{equation}
where
\begin{equation}
M_{x,y} :=\sum_{j,j'=0}^n\Theta(\mu_j,\mu_{j'})
      \psi_j(x)\psi_{j'}(x)\psi_j(y)\psi_{j'}(y) .
\label{eq:54}
\end{equation}

From \eqref{pqpp} we have
$$
\lang p_x^2\rang=\lang (p'_x)^2\rang+\lang \bar p_x^2\rang.
$$
We can further write
\begin{align}
  \label{031012-21}
  & \lang p_x^2\rang -\sum_{y=0}^nM_{x,y} \lang p_y^2\rang
    =G_x^{(n)},\quad
    \mbox{where}\\
&
G_x^{(n)}:=\big(T_--\lang p_0^2\rang \big) M_{0,x} +\lang \bar p_x^2\rang .\notag
\end{align}

\bigskip


\subsection{A lower bound on matrix $[M_{x,y}]$}
 The main result of the present section is the following.
   \begin{proposition}
\label{prop031012-21}
There exists $c_*>0$ such that
\begin{align}
\label{lowerM}
\sum_{x,y=0}^n(\delta_{x,y}-M_{x,y})f_yf_x 
\ge c_*\sum_{x=0}^{n-1} (\nabla f_x)^2, \,\quad\mbox{ for any
  }(f_x)\in\bbR^{n+1},\, n=1,2,\dots.
\end{align}
\end{proposition}
\proof
For any sequence $(f_x)\in\bbR^{n+1}$ we can write
\begin{align}
\label{071504-22}
&\sum_{x,y=0}^n(\delta_{x,y}-M_{x,y})f_yf_x =\sum_{x,y=0}^n \sum_{j,j'=0}^n
 \left(1- \Theta(\mu_j,\mu_{j'})\right)
  \psi_j(x)\psi_{j'}(x)\psi_j(y)\psi_{j'}(y)   f_yf_x\\
&= \sum_{j,j'=0}^n \left(1- \Theta(\mu_j,\mu_{j'})\right)
     \left(\sum_{x=0}^n \psi_j(x)f_x\psi_{j'}(x)\right)^2.\notag
\end{align}
An elementary argument (cf \eqref{eq:47}) shows that there exists $C_*>0$, such that 
$$
1- \Theta(\mu_j,\mu_{j'}) \ge C_*(\mu_j -\mu_{j'})^2
= C_* \left(\la_j-\la_{j'}\right)^2.
$$
Therefore, see \eqref{laps},
\begin{align*}
&\sum_{x,y=0}^n(\delta_{x,y}-M_{x,y})f_yf_x 
\ge C_* \sum_{j,j'=0}^n \left(\la_j-\la_{j'}\right)^2
       \left(\sum_{x=0}^n \psi_j(x)f_x\psi_{j'}(x)\right)^2\\
 &
=C_*\sum_{x,y=0}^n\sum_{j,j'=0}^n \left(\la_j-\la_{j'}\right)^2
        \psi_j(x)\psi_{j'}(x) \psi_j(y)\psi_{j'}(y) f_x f_y\\
&
=2C_*  \sum_{x,y=0}^n\Big\{ \sum_{j=0}^n \la_j^2
  \psi_j^2(x)\delta_{x,y}-\sum_{j,j'=0}^n \la_j
  \la_{j'}\psi_j(x)\psi_{j'}(x)\psi_j(y)\psi_{j'}(y) \Big\}  f_yf_x\\
&
=2C_*\sum_{x,y=0}^n f_yf_x\Big\{
  \delta_{x,y}\langle\Delta_{\rm N}^2 \delta_x,\delta_y\rangle -\langle\Delta_{\rm N}\delta_x,\delta_y\rangle^2\Big\}.
\end{align*}
As before, here $\Delta_{\rm N}$ is the Neumann laplacian.
A careful calculation shows that
$$
\delta_{x,y}\langle\Delta_{\rm N}^2\delta_x,\delta_y\rangle -\langle\Delta_{\rm N} \delta_x,
\delta_y\rangle^2=\langle\Delta_{\rm N} \delta_x,\delta_y\rangle.
$$
Therefore
\begin{align}
&\sum_{x,y=0}^n(\delta_{x,y}-M_{x,y})f_yf_x 
\ge 2C_*\sum_{x=0}^{n-1} (\nabla f_x)^2
\end{align}
and the conclusion of the proposition follows.\qed

\begin{proposition}
\label{prop021012-21}
There exists $C>0$ such that
\begin{equation}
\label{041012-21}
\sup_{x}|G_x^{(n)}| \le \frac{C}{n},\quad n=1,2,\ldots.
\end{equation}
\end{proposition}
\proof
The result is a straightforward consequence of Corollary
\ref{thm021911-21} and Proposition \ref{prop011012-21}.
\qed

From Propositions  \ref{prop031012-21}  and \ref{prop021012-21}  we
immediately conclude the following.
 \begin{corollary}
\label{prop031012-21aa}
There exists $C>0$ such that
\begin{equation}
\label{051012-21}
\sum_{x=0}^{n-1}[\lang p_x^2\rang-\lang p_{x+1}^2\rang]^2 \le
\frac{C}{n+1}\sum_{x=0}^n  \lang p_x^2\rang,\quad n=1,2,\ldots.
\end{equation}
\end{corollary}
\begin{proof}
Using \eqref{lowerM} and then \eqref{031012-21} we get 
  \begin{equation}
    \label{eq:55}
    \begin{split}
    \sum_{x=0}^{n-1}[\lang p_x^2\rang-\lang p_{x+1}^2\rang]^2 \le c_*^{-1}
    \sum_{x,y=0}^n(\delta_{x,y}-M_{x,y}) \lang p_y^2\rang {\lang p_{x}^2\rang} \\
   =\frac1{c_*} \sum_{x=0}^n G^{(n)}_x \lang p_{x}^2\rang \le \frac{C}{c_*(n+1)}\sum_{x=0}^n \lang p_{x}^2\rang
  \end{split}
  \end{equation}
and the conclusion of the corollary follows. 
\end{proof}

\subsection{Upper bound on energy}

The main objective of the present section is the following.
\begin{theorem}
\label{thm011911-21z}
Under assumptions of Theorem \ref{thm-current}, there exists $C>0$ such that 
\begin{equation}
\label{021911-21a}
\frac{1}{n+1}\sum_{x=0}^n\lang {\cal E}_x\rang\le C,\quad n=1,2,\ldots.
\end{equation}
\end{theorem}
Before presenting the proof of Theorem \ref{thm011911-21z}
in Section \ref{sec7.3} we show some auxiliary results.
\begin{proposition}
\label{prop010803-22}
As  $n\to+\infty$ we have 
  \begin{equation}
    \label{eq:32a}
    \sum_{x=0}^n \lang q_x^2 \rang =
   \Big( G_{\omega_0}(0) + o(1)\Big) \sum_{x=0}^n \lang p_x^2 \rang + O(1) \left( \sum_{x=0}^n \lang p_x^2 \rang \right)^{1/2} + O(1).
  \end{equation}
\end{proposition}
\begin{proof}
  Since $\lang \bar q_y^2 \rang$ is of order $n^{2a}$, see
  \eqref{021012-21},
  and $a\le 0$ (cf  \eqref{a-b}) we only need to prove \eqref{eq:32} for
  $\sum_x \lang S^{(q)}_{x,x}\rang $. Using \eqref{eq:69} and then
  \eqref{eq:49} we obtain
  \begin{equation}
    \label{eq:31}
    \begin{split}
      &\sum_x \lang S^{(q)}_{x,x}\rang = \sum_j  \tilde  S^{(q)}_{j,j}
      = \sum_j \frac{1}{\mu_j} \tilde F_{j,j}
      = \sum_y \sum_j \frac{\psi_j(y)^2}{\mu_j} \lang p_y^2 \rang
      + (T_- - \lang p_0^2 \rang ) \sum_j \frac{\psi_j(0)^2}{\mu_j} \\
      &= \sum_y G^n_{\omega_0} (y,y) \lang p_y^2 \rang
      + (T_- - \lang p_0^2 \rang )  G_{n,\omega_0} (0,0),
    \end{split}
  \end{equation}
  where $G^n_{\omega_0}$ is the Green's function of
  $\om_0^2-\Delta_{\rm N}$, see Section \ref{sec2.6}.
  The second term on the utmost right hand side of  \eqref{eq:31} is
  of order $\frac 1n$. Concerning the first term,
thanks to Lemma \ref{lem-ful}, it equals 
  \begin{equation}
    \label{eq:56}
    \left(G_{\omega_0} (0) + o(1) \right)  \sum_{y=0}^{n} \lang
    p_y^2\rang+ \sum_{y=0}^{n} \tilde H^{(n)}(y)\lang
    p_y^2\rang  ,
  \end{equation}
where $|\tilde H^{(n)}(y)|$ satisfies estimate \eqref{042912-21}.
Hence,
  \begin{equation}
    \label{eq:57}
    \begin{split}
   \left|\sum_{y=0}^{n} \tilde H^{(n)}(y)\lang
    p_y^2\rang  \right| &\le  \sum_{y=0}^n |\tilde H_y^{(n)}|\left(\sum_{x=0}^{y-1}|\lang
    p_{x+1}^2\rang-\lang
                 p_x^2\rang|\right)\\
&
+ \lang
    p_0^2\rang \sum_{y=0}^{n}|\tilde H^{(n)}(y)|  .
    \end{split}
  \end{equation}
Thanks to \eqref{042912-21} we have 
$$
 H_*:=\sup_{n\ge1}\sum_{y=0}^{n}|\tilde H^{(n)}(y)|<+\infty.
$$
Thus, using 
\eqref{031911-21} we conclude that the second term on the right hand
side of \eqref{eq:57}
stays bounded, as $N\to+\infty$. The first term can be rewritten by
changing order of summation and then estimated using
the Cauchy-Schwarz inequality and bound \eqref{051012-21}
\begin{align}
\label{030703-22}
  & \sum_{x=0}^n |\lang
    p_{x+1}^2\rang-\lang
                 p_x^2\rang|\left( \sum_{y=x+1}^{n}|\tilde
    H_y^{(n)}|\right) \notag\\
&
   \le   \left\{ \sum_{x=0}^n \left( \sum_{y=x+1}^{n}|\tilde H_y^{(n)}|\right)^2\right\}^{1/2}\left\{ \sum_{x=0}^n [\lang
    p_{x+1}^2\rang-\lang
    p_x^2\rang]^2\right\}^{1/2}\\
  &
    \le H_*\sqrt{n}\left\{ \sum_{x=0}^n [\lang
    p_{x+1}^2\rang-\lang
                 p_x^2\rang]^2\right\}^{1/2}\le C  \left(\sum_{y=0}^n \lang
    p_y^2\rang\right)^{1/2}\notag
\end{align}
and the conclusion of the proposition follows.


\end{proof}

Similarly we obtain
\begin{proposition}
As  $n\to+\infty$ we have 
  \begin{equation}
    \label{eq:32b}
   \begin{split}
  &\sum_{x=0}^n \lang q_x q_{x-1} \rang =
   \left( G_{\omega_0}(1) + o(1)\right) \sum_{x=0}^n \lang p_x^2 \rang
   + O(1) \left( \sum_{x=0}^n \lang p_x^2 \rang \right)^{1/2} + 
   O(1),\\
&\sum_{x=0}^n \lang q_x q_{x+1} \rang =
   \left( G_{\omega_0}(1) + o(1)\right) \sum_{x=0}^n \lang p_x^2 \rang
   + O(1) \left( \sum_{x=0}^n \lang p_x^2 \rang \right)^{1/2} + 
   O(1),\\
&
\sum_{x=0}^n \lang q_{x+1} q_{x-1} \rang =
   \left( G_{\omega_0}(2) + o(1)\right) \sum_{x=0}^n \lang p_x^2 \rang
   + O(1) \left( \sum_{x=0}^n \lang p_x^2 \rang \right)^{1/2} + 
   O(1).
\end{split}
  \end{equation}
\end{proposition}
\begin{proof}
We only prove the first formula of \eqref{eq:32b}. The second and
third ones
follow analogously.
 Using \eqref{eq:49} and \eqref{eq:69} we obtain
  \begin{equation}
\label{010703-22a}
    \begin{split}
      &\sum_{x=0}^n   \lang S^{(q)}_{x,x-1}\rang = \sum_{x=0}^n \sum_{j,j'} \tilde S^{(q)}_{j,j'}
      \psi_j(x) \psi_{j'}(x-1)\\
      &= \sum_{j,j'} \frac{ 2  \Theta(\mu_j,\mu_{j'}) \Upsilon_{j,j'}(1)}{\mu_{j} + \mu_{j'}}
      \sum_{y}  \psi_{j}(y) \psi_{j'}(y)  \lang p_y^2 \rang \\
&
+\Big(
      T_--  \lang p_0^2\rang \Big)\sum_{j,j'} \frac{ 2
        \Theta(\mu_j,\mu_{j'}) \Upsilon_{j,j'} (1)\psi_j(0)\psi_{j'}(0)
      }{\mu_{j} + \mu_{j'}}.
\end{split}
  \end{equation}
Here, we have used the convention   $\lang S^{(q)}_{0,-1}\rang =\lang
S^{(q)}_{0,0}\rang$ and the notation 
\begin{equation}
\label{Ups}
\begin{split}
&\Upsilon_{j,j'}(1):=\frac12\sum_{x=0}^n \big[\psi_j(x)
\psi_{j'}(x-1)+\psi_{j'}(x) \psi_{j}(x-1)\big],\\
& 
\Upsilon_{j,j'}(2):=\frac12\sum_{x=0}^n \big[\psi_j(x-1)
\psi_{j'}(x+1)+\psi_{j'}(x+1) \psi_{j}(x-1)\big].
\end{split}
\end{equation}
A simple application of trigonometric
identities leads to the following formulas.
\begin{lemma}
  For $\ell=1,2$
  \begin{align}
\label{020703-22}
     &\Upsilon_{0,j'}(\ell)=0,\quad j'=1,\ldots,n,\notag\\
   &\Upsilon_{j,j'}(\ell)=\frac{1-(-1)^{j+j'} }{ n+1} \cos\left(\frac{\pi
    j'\ell}{2(n+1)}\right) \cos\left(\frac{\pi
     j\ell}{2(n+1)}\right) ,\quad j,j'=1,\ldots,n,\,j\not=j'\\
     &
       \Upsilon_{j,j}(\ell)
    =\cos\left(\frac{\pi\ell
   j}{n+1}\right) ,\quad j=0,\ldots,n.\notag
  \end{align}
  \end{lemma}
This result implies that the second term on the right hand side of
\eqref{010703-22a} is of order $O(1/n)$.
  The first term equals 
\begin{align*}
&
\sum_{j} \frac{    \Upsilon_{j,j}(1)}{\mu_{j}  }
      \sum_{y}  \psi_{j}^2(y)   \lang p_y^2 \rang +R_n,\quad \mbox{where} \\
&
R_n:=
      \sum_{y}  \rho_n(y) \lang p_y^2
  \rang\quad\mbox{and}\\
&
 \rho_n(y):=\sum_{j\not=j'} \frac{ 2  \Theta(\mu_j,\mu_{j'}) \Upsilon_{j,j'}(1)}{\mu_{j} + \mu_{j'}}
         \psi_{j}(y) \psi_{j'}(y)  .
\end{align*}
Using formula  for $\psi_{j}(y)$, see \eqref{laps}, and  for
$\Upsilon_{j,j}$  we can rewrite the
above expression in the form
\begin{align*}
&
     \tilde g_n(1) \sum_{y}   \lang p_y^2 \rang +R_n+\tilde R_n+o(1) \sum_{y}   \lang p_y^2 \rang,\quad \mbox{where} \\
&
\tilde g_n(1) :=\frac{1}{n+1}\sum_{j} \frac{ 1}{\mu_{j}  }\cos\left(\frac{\pi
   j}{n+1}\right),\\
&
\tilde R_n:=
      \sum_{y}  \tilde \rho_n(y) \lang p_y^2 \rang \quad\mbox{and}\\
&
\tilde \rho_n(y):=\frac{2}{n+1}\sum_{j} \frac{ 1 }{\mu_{j}  } \cos\left(\frac{\pi
   j}{n+1}\right)\cos\left(\frac{\pi
   j(2y+1)}{n+1}\right).
\end{align*}
A simple calculation shows that $\tilde g_n(1) $ is the Riemann sum
approximation of order $1/n$ of the integral defining $G_{\om_0}(1)$,
see \eqref{GR}. An application of Lemma \ref{lm012812-21} yields 
that
 both $| \rho_n(y) |$ and $|\tilde \rho_n(y) |$ satisfy the bound \eqref{042912-21}.
Repeating the estimates done in \eqref{eq:57} and \eqref{030703-22} we
conclude that
$$
|R_n| +|\tilde R_n|\le C  \left(\sum_{y=0}^n \lang
    p_y^2\rang\right)^{1/2},
$$
where the constant $C$ is independent of $n$. The conclusion of the
lemma follows.
\end{proof}

\subsection{Proof of Theorem \ref{thm011911-21z}}
\label{sec7.3}

Using \eqref{eq:32}, \eqref{eq:32b}  and \eqref{D1} we obtain
\begin{equation}
  \label{eq:58}
  \begin{split}
    \sum_{x=0}^n \lang \mathfrak F_x \rang &=
    \sum_{x=0}^n \left[\lang p_x^2 \rang + \lang q_xq_{x+1} \rang + \lang q_xq_{x-1}\rang
      -  \lang q_{x-1}q_{x+1} \rang - (1+\omega^2_0)  \lang q_x^2 \rang
    \right]\\
    &= \left(D + o(1)\right) \sum_{x=0}^n \lang p_x^2 \rang + O(1) \left( \sum_{x=0}^n \lang p_x^2 \rang \right)^{1/2} + O(1).
  \end{split}
\end{equation}
On the other hand, from \eqref{072606-21a} and \eqref{022801-22}, we conclude
\marginpar{\red{Corr.}}
\begin{equation}
  \label{eq:29}
  \begin{split}
    &\sum_{x=0}^n \lang \mathfrak F_x \rang = (n +1) \lang \mathfrak F_0 \rang
    - 2\gamma J_n n\red{(n+1)},\\
&
+\frac{
  n^{a}}{\theta_n}\int_0^{\theta_{n}}\cF(t/\theta_n) \left(
  \bar q_{n-1}(t)-\bar q_{n}(t) \right)\dd t \\
&
\le (n+1) \lang p_0^2 \rang - 2\gamma n^2 J_n +O(n^{2a})\le C' (n+1)
  \end{split}
\end{equation}
that gives the bound
\begin{equation}
  \label{eq:30}
  \frac 1{n+1} \sum_{x=0}^n \lang p_x^2 \rang \le C,\quad n=1,2,\ldots.
\end{equation}
The conclusion of the theorem then follows from the above estimate and
Proposition \ref{prop010803-22}.
\qed

 \subsection{$H^1$ bound on the energy  density}
\label{sec9}

As a direct consequence of \eqref{eq:30} and Corollary \ref{prop031012-21aa} we have 
\begin{corollary}
  \label{cor011312-21}
 There exists $C>0$ such that
\begin{equation}
\label{071512-21}
\sum_{x=0}^{n-1}[\lang p_x^2\rang-\lang p_{x+1}^2\rang]^2 \le C 
\end{equation}
and
 \begin{equation}
     \label{031312-21a}
 \sup_{x=0,\ldots,n}\lang p_x^2\rang_n \le Cn^{1/2},\quad n=1,2,\ldots
 \end{equation}
  \end{corollary}
  \proof Estimate \eqref{071512-21} is obvious in light of \eqref{051012-21}.
   To prove \eqref{031312-21a} note that
   \begin{align*}
 &\lang p_x^2\rang\le \sum_{y=1}^n|\lang p_y^2\rang-\lang
     p_{y-1}^2\rang|+\lang p_0^2\rang\\
     &
       \le \sqrt{n}\left\{\sum_{y=1}^n[\lang p_y^2\rang-\lang
      p_{y-1}^2\rang]^2\right\}^{1/2}+\lang p_0^2\rang
       \le C\sqrt{n} +\lang p_0^2\rang.
   \end{align*}
  \qed


The following result  stregthens    the above estimates.
   \begin{proposition}
\label{prop031012-21a}
There exists $C>0$ such that
\begin{equation}
\label{051012-21a}
\begin{split}
&\sum_{x=0}^{n-1}[\lang p_x^2\rang-\lang p_{x+1}^2\rang]^2 \le
\frac{C}{n+1},\quad n=1,2,\ldots,\\
&
\sup_{x=0,\ldots,n}\lang p_x^2\rang\le C.
\end{split}
\end{equation}
\end{proposition}
\proof
From \eqref{031012-21} and  Propositions \ref{prop031012-21} and \ref{prop011012-21}
there exists $C>0$ such that
\begin{align}
  \label{051012-21bb}
  & \sum_{x=0}^{n-1}\Big(\lang p_x^2\rang-\lang p_{x+1}^2\rang\Big)^2   \le
C\left|\sum_{x=0}^n G_x^{(n)}\lang p_x^2\rang\right|\notag\\
& \le \frac{C}{n+1} +C\sum_{x=0}^n  \lang \bar p_x\rang^2 \lang p_x^2\rang\\
&  \le \frac{C}{n+1} +C \sup_x\lang p_x^2\rang \sum_{x=0}^n  \lang \bar p_x\rang^2
  \le \frac{C}{n+1} +\frac{C}{n+1}  \sup_x\lang p_x^2\rang,
\end{align}
Using the Cauchy-Schwarz inequality we conclude
\begin{align}
  \label{051012-21bd}
& \sup_{x}\lang
  p_x^2\rang\le \lang
  p_0^2\rang +\sum_{x=0}^{n-1}\Big|\lang p_x^2\rang-\lang p_{x+1}^2\rang\Big|\notag\\
&
\le  \lang
  p_0^2\rang +\sqrt{n}\left\{\sum_{x=0}^{n-1}\Big(\lang p_x^2\rang-\lang p_{x+1}^2\rang\Big)^2 \right\}^{1/2}  
\end{align}
Denote
$
D_n:=\sum_{x=0}^{n-1}\Big(\lang p_x^2\rang-\lang p_{x+1}^2\rang\Big)^2.
$
We can summarize the inequalities obtained as follows: there exists
$C>0$ such that
\begin{align}
  \label{051012-21bddd}
&
D_n\le \frac{C}{n+1} +\frac{C}{n+1}  \sup_x\lang
  p_x^2\rang, \\
& \sup_{x}\lang
  p_x^2\rang\le \lang
  p_0^2\rang +\sqrt{n+1}D_n^{1/2} \le\lang
  p_0^2\rang +C+C \sup_x\lang
  p_x^2\rang^{1/2},\notag
\end{align}
for all $n=1,2,\ldots$.
Thus the second estimate of \eqref{051012-21a} follows, which in turn
implies the first estimate of
\eqref{051012-21a}  as well.
\qed

\section{Convergence of the energy density}
\label{sec-converg}

\subsection{Identification of the macroscopic   energy density
  limit. Proof of Theorem \ref{th1}}
Let
$$
e_n(u):=\lang p_x^2\rang,\quad u\in\Big[\frac{x}{n+1},
\frac{x+1}{n+1}\Big),\quad x=0,\ldots,n.
$$
By Proposition \ref{prop031012-21a} and Corollary \ref{thm011911-21z}  we have
$$
\int_0^1e_n^2(u)\dd u \le
\Big(\sup_{x}\lang p_x^2\rang\Big)\Big(\frac{1}{n+1}\sum_{x=0}^n\lang
p_x^2\rang\Big)\le C.
$$

Let $\tilde e_n(u)$ be piecewise linear function obtained by the linear
interpolation between the nodal points
$
P_x:=\Big(\frac{x}{n+1},
\lang p_x^2\rang\Big) $, $x=0,\ldots,n+1$ (here $p_{n+1}=p_n$).
 By Proposition \ref{prop031012-21a}
$$
\int_0^1[\tilde e_n'(u)]^2\dd u= (n+1)\sum_{x=0}^n\Big(\lang p_{x+1}^2\rang-\lang p_{x}^2\rang\Big)^2\le
  C,\quad n=1,2,\ldots.
$$
Therefore $\left(\tilde e_n(u)\right)$ is relatively compact in both
$L^2(0,1)$ and $C[0,1]$. Since
$$
\int_0^1[\tilde e_{n}(u)-e_{n}(u)]^2\dd u\le \frac{1}{n+1}\sum_{x=0}^n\Big(\lang p_{x+1}^2\rang-\lang p_{x}^2\rang\Big)^2\le
\frac{C}{(n+1)^2}
$$
also $\left( e_n(u)\right)$ is relatively compact in
$L^2(0,1)$.

For some subsequence $n'\to+\infty$ (which we still denote by $n$) we
have therefore
$$
\lim_{n\to+\infty}e_n(u)=e_{\rm thm}(u),\qquad \mbox{ in the $L^2$ sense}
$$
and $e_{\rm thm}\in C[0,1]$. {In order to  show the
convergence of $\frac 1n \sum_x \varphi\left(\frac xn \right) \lang
p^2_x\rang$ in \eqref{eq:3} it suffices therefore to prove the following.}
\begin{theorem}
\label{thm012912-21}
We have
\begin{equation}
\label{ethm}
e_{\rm thm}(u)=T(u),\quad u\in[0,1],
\end{equation}
where $T$ is given by \eqref{eq:5}.  
\end{theorem}
\proof 
Since $\lang p_0^2 \rang \to T_-$, compactness of $\left(\tilde e_n(u)\right)$
in $C[0,1]$
implies that $\lim_{u\to 0} e_{\rm thm}(u) = T_-$.
It follows that we only have to verify that
\begin{equation}
  \label{eq:34}
  J\varphi(1) = \frac{D}{4\gamma} \int_0^1 \varphi''(u) e_{\rm thm}(u) du, 
\end{equation}
for any $\varphi\in C^2[0,1]$ such that ${\rm
  supp}\,\varphi\subset (0,1]$ and $\varphi'(1)=0$. Here and below, for
abbreviation sake, we let $J:=J^{b-1/2,b}$, $J_n:=J_n^{b-1/2,b}$ (see
\eqref{eq:38a} and \eqref{051021-05}).

The left hand side of \eqref{eq:34} is given by
\begin{equation}
  \label{eq:35}
  \lim_{n\to\infty} n J_n  \varphi\left(1-\frac 1n\right)  = J\varphi(1).
\end{equation}
On the other hand, from \eqref{eq:33}, for $n$ sufficiently large so $\varphi(1/n)=0$
\begin{equation}
  \label{eq:36}
  \begin{split}
    n J_n \varphi\left(1-\frac 1n\right) &= \sum_{x=1}^{n-2}
    n \left[\varphi\left(\frac {x+1}n\right) - \varphi\left(\frac {x}n\right)\right] \lang j_{x,x+1} \rang\\
    &=- \frac 1{4\gamma} \sum_{x=1}^{n-2}
    n \left[\varphi\left(\frac {x+1}n\right) - \varphi\left(\frac {x}n\right)\right]
    \lang \nabla \frak F_x \rang\\
   &= \frac 1{4\gamma}  \frac 1n \sum_{x=1}^{n-2}
    n^2 \left[\varphi\left(\frac {x+1}n\right) + \varphi\left(\frac {x-1}n\right)
      - 2\varphi\left(\frac {x}n\right)\right]
    \lang \frak F_x \rang \\
    &= \frac 1{4\gamma}  \frac 1n \sum_{x=0}^{n-2}
    \varphi''\left(\frac xn\right)  \lang \frak F_x \rang + R_n
  \end{split}
\end{equation}
with
\begin{equation}
  \label{eq:59}
  |R_n| \le \frac{C}{n^2} \sum_x |\lang \frak F_x \rang| \mathop{\longrightarrow}_{n\to\infty} 0.
\end{equation}
 following from \eqref{021911-21a}.

Then we are left to prove that
\begin{equation}
  \label{eq:60}
 \lim_{n\to\infty} \frac 1n \sum_{x=0}^{n-2} \varphi''\left(\frac xn\right)
  \left(\lang \frak F_x \rang - D \lang p_x^2 \rang \right) = 0.
\end{equation}
This will follows directly if we prove the following
\begin{proposition}
  \label{prop-loceq}
For any test function  $\varphi\in C^2([0,1])$ such that ${\rm
  supp}\,\varphi\subset (0,1]$ and $\varphi'(1)=0$ we have
\begin{equation}
  \label{eq:61}
\lim_{n\to\infty} \frac 1n \sum_{x=0}^{n-2} \varphi''\left(\frac xn\right)
  \left(\lang q_xq_{x+\ell } \rang - G_{\omega_0}(\ell) \lang p_x^2 \rang \right) = 0,
  \qquad \ell= 0, 1 ,2.
\end{equation}
\end{proposition}

\begin{proof}
  {By virtue of \eqref{021012-21} we have
 \begin{equation}
    \label{eq:61ab}
    \lim_{n\to+\infty}\frac{1}{n} \sum_{x=0}^{n}
    \varphi\left(\frac{x}{n}\right)  \lang
      \bar q_x \bar q_{x+\ell}\rang=0.
  \end{equation}
It suffices therefore to prove that 
\begin{equation}
    \label{eq:61ap}
    \lim_{n\to+\infty}\frac{1}{n} \sum_{x=0}^{n}
    \varphi\left(\frac{x}{n}\right) \int_0^{t}  \Big\{\lang
       q_x'  q_{x+\ell}'\rang-  G_{\omega_0}(\ell)
   \lang p_x^2\rang\Big\}\dd s=0.
  \end{equation}}
We first prove \eqref{eq:61ap} for $\ell=0$. By   \eqref{eq:69} we have 
  \begin{equation}
    \label{eq:62}
    \begin{split}
      {\lang (q_x')^2 \rang} = \sum_{j,j'} {\tilde S^{(q)}_{j,j'}}
      \psi_j(x) \psi_{j'}(x)   
      = H_{n}(x) +O\Big(\frac{1}{n}\Big) ,
    \end{split}
\end{equation}
with
\begin{equation}
  \label{eq:63}
 \begin{split}
& H_{n}(x) =
\sum_{y=0}^n \sum_{j,j'} \Phi\left(\frac{j}{n+1}, \frac{ j'}{n+1}\right)
  \psi_j(y) \psi_{j'}(y) \psi_j(x) \psi_{j'}(x)  \lang p_y^2 \rang
\\
&
= H_n^1(x)+\frac{1}{n+1} H_n^2(x)+\frac{\Phi\left(0,0\right)}{(n+1)^2}\sum_{y=0}^n \lang p_y^2 \rang.
\end{split}
\end{equation}
Here $\Phi\left(\frac{j}{n+1}, \frac{ j'}{n+1}\right)
= \frac{2\Theta(\mu_j,\mu_{j'})}{\mu_j+\mu_{j'}}$ and 
$
H_n^j(x):= \sum_{y=0}^n K^{(n)}_j(x,y)  \lang p_y^2 \rang,
$
with
\begin{equation}
  \label{eq:70}
  \begin{split}
&K_1 ^{(n)} (x,y) :=\frac{1}{(n+1)^2}\sum_{j,j'=0}^n \Phi\left(\frac{j}{n+1}, \frac{j'}{n+1}\right)
    \Big[\cos\left(\frac{\pi j(y-x)}{n+1}\right)+\cos\left(\frac{\pi
        j(y+x+1)}{n+1}\right) \Big] \\
    &\times  \Big[\cos\left(\frac{\pi
        j'(y-x)}{n+1}\right)+\cos\left(\frac{\pi
        j'(y+x+1)}{n+1}\right) \Big],\\
&
  K_2 ^{(n)} (x,y) :=-\frac{1}{n+1}\sum_{j=0}^n \Phi\left(\frac{j}{n+1},0\right)
    \Big[\cos\left(\frac{\pi j(y-x)}{n+1}\right)+\cos\left(\frac{\pi
        j(y+x+1)}{n+1}\right) \Big] .
  \end{split}
\end{equation}
Using Lemma \ref{lm012812-21}
we obtain that
\begin{equation}
\label{023112-21aa}
\begin{split}
&|K_j^{(n)}(x,y)|\le
C\left(\frac{1}{\chi^2_n((x-y)/2)}+\frac{1}{\chi^2_n((x+y)/2)}\right), \quad x,y=0,\ldots,n,\,n=1,2,\ldots
\end{split}
\end{equation}
for $j=1,2$.
In particular the above estimate implies that
\begin{equation}
\label{023112-21ac}
K_{j,*}:=\sup_{x,n}\sum_{y=0}^n|K^{(n)}_j(x,y)|<+\infty,\quad j=1,2.
\end{equation}
In consequence, by virtue of \eqref{051012-21a},
\begin{equation}
  \label{eq:63a}
  |H_{n}^j(x)| \le \sum_{y=0}^n| K^{(n)}_j(x,y)|  \lang p_y^2 \rang\le K_{j,*}\sup_{y,n} \lang p_y^2 \rang=:H_{j,*}<+\infty
\end{equation}
and the term corresponding to $H_{n}^2(x)$ is negligible, as $n\to+\infty$.

Choose $\delta\in(0,1)$ sufficiently small, so that $\varphi(u)=0$,
when $u\in(0,\delta)$.
Then, there exist $C>0$ and
{$n_0\ge 1$} such that
\begin{align*}
&\frac 1n \left|\sum_{x\in(0,\delta n)\cup((1-\delta)n,n)}\varphi''\left(\frac xn\right)
                 {\lang (q_x')^2 \rang} \right|\le \frac 1n\sum_{ x \in ((1-\delta)n,n)}
                 \left|\varphi''\left(\frac xn\right)\right|
  \left(|H_n(x)|+O\left(\frac1n\right)\right)\\
&
\le C\|\varphi''\|_\infty\delta,\quad {n=n_0+1,n_0+2,\ldots.}
\end{align*}
For  $\delta n\le x\le (1-\delta)n$ inequality \eqref{023112-21aa} implies that
there exists $C>0$ such that  $K_1 ^{(n)} (x,y)=\bar K_1 ^{(n)}
(x,y)+O(1/n^2)$, where
\begin{equation}
  \label{eq:70b}
\bar K_1 ^{(n)} (x,y) :=\frac{1}{4(n+1)^2}\sum_{j,j'=-n-1}^n \Phi\left(\frac{j}{n+1}, \frac{j'}{n+1}\right)
    \cos\left(\frac{\pi j(y-x)}{n+1}\right) \cos\left(\frac{\pi
        j'(y-x)}{n+1}\right),
\end{equation}
and, by Lemma \ref{lm012812-21}, there exists $C>0$ such that
\begin{equation}
\label{023112-21b}
|\bar K_1 ^{(n)} (x,y)|\le
\frac{C}{1+(x-y)^2},\quad y=0,\ldots,n,\,x\in(\delta n,(1-\delta)n)
\end{equation}
for $n=1,2,\ldots.$
To prove \eqref{eq:61} it suffices therefore to show that for any
$\delta\in(0,1/2)$ we have
\begin{equation}
  \label{eq:61d}
\lim_{n\to\infty} \frac 1n \sum_{\delta n\le x\le (1-\delta)n}  \varphi''\left(\frac xn\right)
  \left(\bar H_n^1(x) - G_{\omega_0}(0) \lang p_x^2 \rang \right) = 0,
\end{equation}
where $\bar H_n^1(x):= \sum_{y=0}^n \bar K^{(n)}_1(x,y)  \lang p_y^2 \rang$.
Using Cauchy-Schwarz inequality,  estimates \eqref{023112-21b} and \eqref{051012-21a} we conclude, that
\begin{equation}
\label{011403-22}
  \begin{split}
&\sum_{y=0}^n \Big|\lang p_y^2\rang-\lang p_x^2\rang\Big|
|\bar K^{(n)}_1({x,y})|\le
\sum_{y=0}^n |y-x|^{1/2}\left|\sum_{z=0}^{n-1}
\left(\lang p_{z+1}^2\rang-\lang p_z^2\rang\right)^2\right|^{1/2}|\bar K^{(n)}_1({x,y})|\\
&\le \frac{C}{n^{1/2}} \sum_{y=0}^n |y-x|^{1/2} |\bar K^{(n)}_1({x,y})|\le
\frac{C'}{n^{1/2}}
\end{split}
 \end{equation}

Therefore
  \begin{equation*}
\begin{split}
&\frac{1}{n+1} \sum_{\delta n\le x\le (1-\delta)n} \varphi''\left(\frac{x}{n+1}\right)   \sum_{y=0}^n \lang p_y^2\rang
 \bar K^{(n)}_1({x,y})
\\
&
  = \frac{1}{n+1} \sum_{\delta n\le x\le (1-\delta)n}  \varphi''\left(\frac{x}{n+1}\right) \lang p_x^2\rang
  \sum_{y=0}^n \bar K^{(n)}_1({y,x})+ O\left(\frac 1{\sqrt n}\right).\label{eq:71}
\end{split}
\end{equation*}
Another application of Lemma \ref{lm012812-21} allows us to conclude that
\begin{equation}
\label{023112-21a}
|\sum_{y=0}^n\bar K^{(n)}_1({y,x})-\sum_{y=x-n-1}^{x+n}\bar K^{(n)}_1({y,x})|\le
\frac{C}{n^2},\quad     \delta n\le x \le (1-\delta)n,
\end{equation}
Using elementary trigonometric identities we conclude that
$$
 2\sum_{z=-n-1}^n \cos\left(\frac{\pi j z}{n+1}\right) \cos\left(\frac{\pi
        j'z}{n+1}\right)=(n+1)\Big(\delta_{j,-j'}+\delta_{j,j'}\Big).
$$
Hence, by virtue of Lemma \ref{lem-ful}, we get
\begin{equation}
\begin{split}
&\sum_{y=x-n-1}^{x+n}\bar K^{(n)}_1({y,x})=
\frac{1}{n+1}\sum_{j=0}^n  \Phi\left(\frac{j}{n+1}, \frac{ j}{n+1}\right) \\
&= \frac{1}{n+1}\sum_{j=0}^n  \frac 1{\mu_j} = G_{\omega_0}(0) + o(1),
\end{split}
\label{eq:72}
\end{equation}
for $\delta n\le x \le (1-\delta)n$.  We have shown therefore that
\begin{align*}
\frac 1n \sum_{\delta n\le x\le (1-\delta)n}  \varphi''\left(\frac xn\right)
   \bar H_n^1(x) = \frac 1n G_{\omega_0}(0) \sum_{\delta n\le x\le (1-\delta)n}  \varphi''\left(\frac xn\right) \lang p_x^2 \rang +o(1)
\end{align*}
and \eqref{eq:61d} follows. This ends the proof of \eqref{eq:61} for $\ell=0$.

\bigskip

 \subsubsection*{Proof for $\ell \neq 0$}.
  The proof is similar to the previously considered case, so we only
  sketch it.
  As before, by \eqref{eq:69} we have
  \begin{equation}
    \label{eq:62}
    \begin{split}
      {\lang q_x' q_{x+\ell}' \rang} =
      \sum_{j,j'} \tilde S^{(q)}_{j,j'} \psi_j(x) \psi_{j'}(x+\ell) + O(1/n)
      = H_{n,\ell}(x) + O(1/n),
    \end{split}
  \end{equation}
with
\begin{equation}
  \label{eq:63}
 \begin{split}
 &
H_{n,\ell}(x) = \sum_{y=0}^n \sum_{j,j'} \frac{2\Theta(\mu_j,\mu_{j'})}{\mu_j+\mu_{j'}}
  \psi_j(y) \psi_{j'}(y) \psi_j(x) \psi_{j'}(x+\ell)  \lang p_y^2 \rang
 \\
&
 = \sum_{y=0}^n \bar K^{(n,\ell)}({x,y})  \lang p_y^2 \rang+O\left(\frac1n\right).
\end{split}
\end{equation}
Here   
\begin{equation}
  \label{eq:70}
   \bar  K^{(n,\ell)}({x,y}):=\frac{1}{4(n+1)^2}\sum_{j,j'=-n-1}^n \Phi\left(\frac{j}{n+1}, \frac{j'}{n+1}\right)
     \cos\left(\frac{\pi j(y-x)}{n+1}\right) \cos\left(\frac{\pi j'(y-x-\ell)}{n+1}\right).
\end{equation}
Using Lemma \ref{lm012812-21}, in a manner similar to what we have
done in the argument leading up  to  \eqref{eq:61d},
we conclude that   it suffices   to show that for any
$\delta\in(0,1/2)$ we have
\begin{equation}
  \label{eq:61e}
\lim_{n\to\infty} \frac 1n \sum_{\delta n\le x\le (1-\delta)n}  \varphi''\left(\frac xn\right)
  \left(\bar H_{n,\ell}(x)  - G_{\omega_0}(\ell) \lang p_x^2 \rang \right) = 0,
\end{equation}
with $\bar H_{n,\ell}(x) := \sum_{y=0}^n \bar K^{(n,\ell)}({x,y})  \lang p_y^2 \rang$.
Furthermore, thanks to estimates analogous to those leading up to
\eqref{023112-21a}, we get
\begin{equation}
\label{023112-21a}
|\sum_{y=0}^n  \bar  K^{(n,\ell)}({y,x})-\sum_{y=x-n-1}^{x+n}\bar  K^{(n,\ell)}({y,x})|\le
\frac{C}{n^2},\quad     \delta n\le x \le (1-\delta)n.
\end{equation}
Using elementary trigonometric identities we conclude that
$$
 2\sum_{z=-n-1}^n \cos\left(\frac{\pi j z}{n+1}\right) \cos\left(\frac{\pi
        j'(z-\ell)}{n+1}\right)=(n+1)\cos\left(\frac{j\ell }{n+1} \right)\Big(\delta_{j,-j'}+\delta_{j,j'}\Big).
$$
Hence,  
\begin{equation}
\begin{split}
&\sum_{y=x-n-1}^{x+n}\bar K^{(n)}_1({y,x})=
\frac{1}{n+1}\sum_{j=0}^n  \Phi\left(\frac{j}{n+1}, \frac{ j}{n+1}\right) \cos\left(\frac{j\ell }{n+1} \right) \\
&= G_{\omega_0}(\ell) + o(1),
\end{split}
\label{eq:72}
\end{equation}
for $\delta n\le x \le (1-\delta)n$.  We have shown therefore that
\begin{align*}
\frac 1n \sum_{\delta n\le x\le (1-\delta)n}  \varphi''\left(\frac xn\right)
   \bar H_{n,\ell}(x) = \frac 1n G_{\omega_0}(\ell) \sum_{\delta n\le x\le (1-\delta)n}  \varphi''\left(\frac xn\right) \lang p_x^2 \rang +o(1)
\end{align*}
and \eqref{eq:61e} follows. This ends the proof of \eqref{eq:61} for
$\ell\not=0$.

{Finally, to finish the proof of Theorem \ref{th1}
we show the following result, that is a form of the equipartition
property of the energy.
\begin{lemma}
\label{cor012912-21}
  Suppose that $\varphi\in C^1[0,1]$. Then,
  \begin{equation}
    \label{011312-21}
\lim_{n\to+\infty}\frac{1}{n+1}\sum_{x=0}^n \varphi\left(\frac{x}{n+1}\right)\Big(\lang p_x^2\rang-\lang r_x^2\rang-\om_0^2 \lang q_x^2\rang\Big)=0.
\end{equation}
Here $r_x:=q_x-q_{x-1}$, $x=1,\ldots,n$ and $r_0:=0$.
  \end{lemma}
  \proof
After a simple calculation we obtain
\begin{equation}
  \label{042606-21aa}
  \lang p_x^2\rang =   \lang \om^2_0q_x^2 -(\Delta q_{x}) q_x\rang 
  -\frac{\delta_{x,n}n^a}{\theta_n }\int_0^{\theta_n}\bar
  q_x(t)\cF(t/\theta)\dd t
\end{equation}
 for $x=0,\ldots,n$.
 Therefore, by \eqref{022801-22} (see also \eqref{In})
\begin{equation}
  \label{081312-21}
  \begin{split}
&\frac{1}{n+1}\sum_{x=0}^n \varphi\left(\frac{x}{n+1}\right)\Big(\lang p_x^2\rang-\lang r_x^2\rang-\om_0^2 \lang q_x^2\rang\Big)\\
&
=\frac{1}{n+1}\sum_{x=0}^{n}\varphi\left(\frac{x}{n+1}\right)\lang r_x
q_x\rang
-\frac{1}{n+1}\sum_{x=1}^{n}\varphi\left(\frac{x-1}{n+1}\right)\lang r_x
                             q_{x-1}\rang
\\
&
-\frac{1}{n+1}\sum_{x=0}^{n}\varphi\left(\frac{x}{n+1}\right)\lang r_x^2\rang+O\left(\frac{1}{n^{1-a}}\right).
\end{split}
\end{equation}
Since $r_0=0$ the last expression equals
\begin{align*}
  &
\frac{1}{n+1}\sum_{x=1}^{n}\Big[\varphi\left(\frac{x}{n+1}\right) -\varphi\left(\frac{x-1}{n+1}\right)\Big]\lang r_x
                             q_{x-1}\rang+O\left(\frac{1}{n^{1-a}}\right)=O\left(\frac{1}{
     n}\right),
\end{align*}
by virtue of \eqref{021911-21a}.
This ends the proof of \eqref{011312-21}.\qed}

\end{proof}

\section{Vanishing time variance of the kinetic energy}
\label{sec:vanish-time-vari}

A natural question is about the time averaging of the energy functional.
We prove   that the time variance of the average kinetic energy
vanishes as $n\to\infty$. We consider only the case when $b=0$ and
$a^{-1/2}$    in the dynamics described by \eqref{eq:flip} --
\eqref{Fnt}.
\begin{theorem}
  \label{thm:var}
Under the assumption stated above there exists a constant
$C>0$ such that  
  \begin{equation}
    \label{eq:91}
  {    \sum_{x=0}^n \frac 1\theta
    \int_0^\theta \left(\bar{p_x^2}(t) - \lang p_x^2\rang \right)^2
    \dd t  
  \le \frac{C}{n^2},\quad n=1,2,\ldots.}
  \end{equation}
\end{theorem}

\begin{proof}
  From \eqref{P2c} we get
  \begin{equation}
    \label{eq:92}
    \bar{p_x^2}(t) = T_- M_{x,0} + \sum_{x'=1}^n \int_0^\theta
    {\frak g}_{x,x'}(s) \bar{p_{x'}^2}(t-s) \dd s
    + \bar{p_x}^2(t),
  \end{equation}
where  ${\frak g}_{x,x'}(s)$ is defined in \eqref{011504-22}.
Averaging over the $t$ variable we get
  \begin{equation}
    \label{eq:92aa}
    \lang
    p_x^2\rang = T_- M_{x,0} + \sum_{x'=1}^n \int_0^\theta
    {\frak g}_{x,x'}(s) \dd s\lang
    p_{x'}^2\rang
    +\lang
   \bar{p_x}^2\rang .
  \end{equation}
  Then, denoting
  \begin{align*}
&V_x(t) :=   \bar{p_x^2}(t) - \lang
    p_x^2\rang  ,\\
&
v_x(t):=  \bar{p_x}^2(t)-\lang
   \bar{p_x}^2\rang,
\end{align*}
we can write
 \begin{equation}
    \label{eq:92ab}
V_x(t)=  \sum_{x'=1}^n \int_0^\theta
    {\frak g}_{x,x'}(s) V_{x'}(t-s) \dd s+ v_x(t).
  \end{equation}

For $m\in\bbZ$ define 
\begin{equation}
\label{Mxym}
\begin{split}
&M_{x,y}(m):=\int_0^\theta
    {\frak g}_{x,y}(s)e^{-2\pi i ms/\theta} \dd s \\
&=4\ga\int_0^{+\infty}e^{-2\pi i
  ms/\theta}\Big([e^{-As}]_{x+n+1,y+n+1}\Big)^2\dd s
\end{split}
\end{equation}
(cf \eqref{011504-22}) and
$$
\tilde V_x(m):=\frac{1}{\theta}\int_0^{\theta}e^{-2\pi  \ii m
  t/\theta}V_x(t)\dd t\qquad\mbox{and}\qquad \tilde v_x(m):=\frac{1}{\theta}\int_0^{\theta}e^{-2\pi  \ii m
  t/\theta}v_x(t)\dd t.
$$
Note that obviously $\tilde V_x(0)=0$ and $M_{x,y}(0)=M_{x,y}$, see \eqref{eq:51}. From \eqref{eq:92ab} we get
\begin{equation}
    \label{021404-22m}
\tilde V_x(m)=  \sum_{x'=1}^n M_{x,x'}(m) \tilde V_{x'}(m)+ \tilde v_x(m).
  \end{equation}
Multiplying both sides by $\tilde V^\star_x(m)$ and summing over $x$
we get
\begin{equation}
    \label{021404-22m}
\sum_{x=0}^n |\tilde V_x(m)|^2=  \sum_{x=0}^n \sum_{x'=1}^n
M_{x,x'}(m) \tilde V_{x'}(m) \tilde V^{\star}_{x}(m)+ \sum_{x=0}^n\tilde v_x(m) \tilde V^{\star}_{x}(m).
  \end{equation}
Hence,
\begin{equation}
    \label{031404-22}
\begin{split}
&\sum_{x,x'=0}^n \Big(\delta_{x,x'}-
M_{x,x'}(m) \Big)\tilde V_{x'}(m) \tilde V^{\star}_{x}(m)\\
&
={ - \tilde V_{0}(m) \sum_{x=0}^n 
M_{x,0}(m) \tilde V^{\star}_{x}(m)} + \sum_{x=0}^n\tilde v_x(m) \tilde V^{\star}_{x}(m).
\end{split}
  \end{equation}
We have the following.
\begin{lemma}
\label{lm011504-22}
There exists a constant ${\frak C}>0$, such that
\begin{align}
\label{021504-22}
\Big|\sum_{x,y=0}^n(\delta_{x,y}-M_{x,y}(m))f_y^\star f_x\Big|\ge  
    {\frak C} \sum_{x=0}^n |f_x |^2,\quad (f_0,\ldots,f_n)\in\mathbb C^{n+1}
\end{align}  
for $m\not=0$ and $n=1,2,\ldots$.
\end{lemma}
The lemma is shown in Appendix \ref{appC}, where we also prove the following
\begin{lemma}
\label{lm021504-22}
We have
\begin{align}
\label{041504-22}
{\frak M}:=\sup_m \left\{ \sum_{x=0}^n
  |M_{x,0}(m)|^2\right\}^{1/2}<+\infty.
\end{align}  
\end{lemma}

{\begin{lemma}
\label{lm011904-22}
There exists a constant ${\frak v}_*>0$ such that
\begin{align}
\label{011904-22}
 \left\{ \sum_{x=0}^n \lang  v_x^2\rang\right\}^{1/2}\le \frac{{\frak v}_*}{n},
\quad n=1,2,\ldots  
\end{align}  
\end{lemma}}

\begin{lemma}
\label{lm031504-22}
There exists a constant ${\frak V}_*>0$ such that
\begin{align}
\label{061504-22}
  \lang  V_0^2\rang^{1/2}\le \frac{{\frak V}_*}{n^{1/2}}
  \left\{ \sum_{x=0}^n \lang  V_x^2\rang\right\}^{1/4}
\end{align}  
\end{lemma}
The lemmas are shown in Appendix \ref{appC}.

We show how to apply these to
finish the proof of the theorem.

From \eqref{031404-22}, \eqref{021504-22} and the Cauchy-Schwarz
inequality we conclude that {
\begin{equation}
    \label{031504-22}
\begin{split}
  & {\frak C} \sum_{x=0}^n |\tilde V_x(m)|^2\le
  \left[ |\tilde V_0(m)|\left\{ \sum_{x=0}^n |M_{x,0}(m)|^2\right\}^{1/2}
    + \left\{ \sum_{x=0}^n |\tilde v_{x}(m)|^2\right\}^{1/2} \right]
    \left\{ \sum_{x=0}^n |\tilde V_x(m)|^2\right\}^{1/2},
\end{split}
\end{equation}
i.e.
\begin{equation}
  \label{eq:25}
  \begin{split}
  & {\frak C}^2 \sum_{x=0}^n |\tilde V_x(m)|^2\le
  \left[ {\frak M} |\tilde V_0(m)|
    + \left\{ \sum_{x=0}^n |\tilde v_{x}(m)|^2\right\}^{1/2} \right]^2
  \le 2 {\frak M}^2  |\tilde V_0(m)|^2 + 2 \sum_{x=0}^n |\tilde v_{x}(m)|^2.
\end{split}
\end{equation}
Summing over $m$ and using \eqref{061504-22} together with \eqref{011904-22} we obtain
\begin{equation}
  \label{eq:39}
  \begin{split}
  {\frak C}^2  \sum_{x=0}^n \lang  V_x^2\rang \le
  2  {\frak M}^2  \lang V_0^2\rang + 2 \sum_{x=0}^n \lang v_{x}^2\rang
  \le  \frac{2  ({\frak M}{\frak V}_*)^2}{n}
  \left\{ \sum_{x=0}^n \lang  V_x^2\rang\right\}^{1/2} + \frac{{\frak v}_*^2}{n^2}\\
  \le  \frac{({\frak M}{\frak V}_*)^4}{{\frak C}^2 n^2} +
  \frac {{\frak C}^2 }2 \sum_{x=0}^n \lang  V_x^2\rang + \frac{{\frak v}_*^2}{n^2},
\end{split}
\end{equation}
and \eqref{eq:91} follows.
}

\end{proof}

\section{Concluding remarks}
\label{sec:conclusions}

In this article we studied the energy transport
in the {periodic} state of a
pinned harmonic chain with bulk dynamics perturbed by a random flip
of the signs of the velocities. Work was done on the system by a
periodic forcing {acting} on the right
hand side of the chain and the heat was absorbed   by a heat bath
coupled to the system via a Langevin stochastic
thermostat at temperature $T_-$ on the left.
The asymptotic temperature profile \eqref{eq:5} should be seen as
the stationary solution of the heat equation:


\begin{equation}
    \label{eq:5ns}
    \begin{split}
      \partial_t T(t,u) &= \frac {D}{4\gamma} \partial_u^2 T(t,u) \qquad u\in (0,1)\\
      T(t,0) &= T_-, \quad  \partial_u T (t,1) = -\frac{4\gamma   J}{D},
      \quad T(0,u) = T_0(u),
    \end{split}
  \end{equation}
  where $  J$ is given by \eqref{051021-05}.
  In a companion work \cite{klo2} we prove that,
  {under proper conditions on the initial distribution,}
  for any {compactly supported}  test function $\varphi\in C([0,+\infty)\times[0,1])$,
  \begin{equation}
    \label{eq:3ns}
    \lim_{n\to\infty} \frac 1n \sum_{x=0}^n
    \int_{0}^{+\infty}\varphi\left(t,\frac xn \right) \mathcal E_x(n^2
    t)\dd t
    = \int_{0}^{+\infty}\int_0^1 \varphi (t,u) T(t,u) \dd t\dd u,
    \quad\mbox{in probability,}
  \end{equation}
  with $T(t,u)$  the solution of \eqref{eq:5ns}. 
  Notice the diffusive rescaling of space and time.

  {The diffusion  coefficient appearing in \eqref{eq:5ns} (defined in
  \eqref{eq:10}) has been computed in a different way in \cite[Theorem
    formula (74)]{BOS10} and \cite[Theorem
  3.2, formula (3.21)]{jko15}. It turns out that it equals to the
  diffusivity of a phonon performing a random walk on 
  the integer lattice
 with random
 scattering generated by the   noise. As a result
  \begin{equation}
    D=2\int_{\bbT}\left(\frac{\om'(k)}{2\pi}\right)^2dk
    \label{eq:41}
\end{equation}
where $\om'(k)/(2\pi)$ is a group velocity of a phonon of frequency
$k$ belonging to the one dimensional unit torus $\bbT$, that 
is the interval $[-1/2,1/2]$ with identified endpoints. Here
$\om(k)=\sqrt{\om_0^2+4\sin^2(\pi k)}$ is the dispersion relation of
the harmonic lattice considered in the present paper. In fact, if we
consider a more general type of noise that allows to scatter the
phonons of given frequency $k$, with the total scattering kernel
$R(k)$ (in the case of the flip noise $R(k)\equiv 1$) we would have
$$
D=2\int_{\bbT}\left(\frac{\om'(k)}{2\pi}\right)^2\frac{dk}{R(k)}.
$$}

As noted before the velocity reversals introduced in the dynamics serve the purpose of
making the heat conductivity finite. In their absence the harmonic crystal has an infinite
conductivity \cite{RLL67}.
The velocity reversals are thus an idealized substitute for the anharmonicities,
impurities and other defects which scatter  phonons and produce a
finite conductivity and {establishes} the validity of Fourier's law
in real solids.

An alternative way of modeling anharmonicity 
for achieving a finite conductivity is the introduction of
"self consistent" reservoirs. This was
introduced in \cite{BRV70,BRV75}  and fully analyzed in \cite{bll}
for describing the heat flow in a harmonic crystal in contact
with two heat reservoirs at different temperatures and  
no external force. In that model one introduces, in addition
to the external 
reservoirs, also "internal" {Langevin}  reservoirs for each particle.
Letting $T_x$ be the temperature of the reservoir at position $x=0,\ldots,n$,
with the same coupling $\gamma$ as used here,
Eq. \eqref{eq:qdynamicsbulk-av} remains unchanged while Eq. \eqref{eq:current}
will have an extra
term, $2\gamma(T_x-p_x^2)$, on its right hand side.
Solving for the periodic first and second moments of the system the internal
$T_x, x=1,\dots, n-1,$ are then determined
by the requirement that the time average of the internal heat flux, given by
the term in the square bracket above, vanishes in the stationary state.
{As a result,} there is no contribution to the average current from these
internal reservoirs and the limiting macroscopic behavior is the same
{as in a corresponding  velocity flip model. 
This approach can be
  modified 
  by considering a periodic, instead of a constant,
self-consistent temperature profile. This will make the dynamics of
the first and second moments of the   position and momenta variables
identical with that of the flip model.}
An important property of the self consistent reservoir model is that the periodic
measure is Gaussian, consequently it is
determined by the first and second moments already computed here.

{Various possible extension of the present model are presented in the review article
  \cite{klos}: forcing acting on a particle in the bulk, unpinned dynamics,
  higher dimensional lattice, anharmonic interactions.}

 \appendix

\section{The proof of Theorem \ref{periodic}}
\label{appB}

Since the result does not depend on the scaling factor $n^a$,
standing by the force
${\cal F}(t)$, and the period size $\theta_n$ we assume that $a=0$ and $\theta_n=\theta$. Given a Borel probability measure $\mu$ on $\bbR^{2(n+1)}$ (see \eqref{eq:1})
we denote by $\big(\mathbf q_{\mu}(t), \mathbf p_{\mu}(t)\big)$ the
solution of \eqref{eq:flip}--\eqref{eq:pbdf}) such that  $\big(\mathbf q_{\mu}(0), \mathbf p_{\mu}(0)\big)$
is distributed according to $\mu$. Denote then by  $\big(\bar{\mathbf
  q}_{\mu}(t), \bar{\mathbf p}_{\mu}(t)\big)$ and $C_{\mu}(t)$ the
vector of averages and matrix of the mixed second moments of the solution, correspondingly. They are
defined by formulas \eqref{eq:6} and a $2\times 2$ block matrix
$$
 C _{\mu} (t)=\left[
  \begin{array}{ll}
    C^{(q)}_{\mu} (t)&C^{(q,p)}_{\mu} (t)\\
 C^{(p,q)}_{\mu} (t)   &C^{(p)}_{\mu} (t)
  \end{array}
\right]
$$
Each block is an $(n+1)\times(n+1)$ matrix
\begin{align*}
&C^{(q)}_{\mu} (t)=[  q_x(t) q_{x'}(t) ]_{x,x'=0,\ldots,n}\quad
C^{(p)}_{\mu} (t)=\bbE[  p_x(t) p_{x'}(t) ]_{x,x'=0,\ldots,n},\\
&C^{(q,p)}_{\mu} (t)=\bbE[   q_x(t) p_{x'}(t)]_{x,x'=0,\ldots,n}
\end{align*}
and $C^{(p,q)}_{\mu} (t)=[C^{(q,p)}_{\mu} (t)]^T$, where the initial state is taken to be $\mu$.

{By similar calculation as done in \eqref{eq:covev},}
their evolution is described by the system of linear differential
equations with periodic forcing
\begin{equation}
\label{010501-22}
\begin{split}
&\frac{\dd }{\dd t}\left(
  \begin{array}{c}\bar{\bf q}_{\mu}(t)\\
    \bar{\bf p}_{\mu} (t)
    \end{array}\right)=-A \left(
  \begin{array}{c}\bar{\bf q}_{\mu} (t)\\
    \bar{\bf p}_{\mu} (t)
    \end{array}\right)+ 
\;\cF(t/\theta){\rm e}_{p,n+1},\\
&
\frac{\dd }{\dd t}C_{\mu}(t)=-A C_{\mu}(t)-
C_{\mu}(t)A^T+\Sigma_2\big({\bf c}_{2,\mu}(t)\big)+  \cF(t/\theta) F(t),
\end{split}
\end{equation}
where 
\begin{equation}
\label{c2C}
{\bf c}_{2,\mu}(t)=\left(
  \begin{array}{c}
  C^{(p)}_{0,0,\mu}(t)  \\
    \vdots\\
     C^{(p)}_{n,n,\mu}(t)
  \end{array}
    \right) 
  \end{equation}
and 
\begin{equation}
\label{Ft}
                              F(t):=\left[
  \begin{array}{cc}
  0&   \bar {\bf q}_\mu(t) \otimes {\rm e}_{p,n+1} \\
&\\
     {\rm e}_{p,n+1}\otimes\bar {\bf q}_\mu (t)  & \ {\rm
                                                   e}_{p,n+1}\otimes
                                                   \bar {\bf p}_\mu
                                                   (t)+\bar {\bf
                                                   p}_\mu (t) \otimes
                                                   {\rm e}_{p,n+1} 
  \end{array}
\right]
 \end{equation}
Here
${\rm e}_{p,n+1}$ and $\Sigma_2$ are defined in  \eqref{epq} and
\eqref{S2} respectively. Suppose that we are given  a vector
$\bar X\in \bbR^{2(n+1)}$ and a symmetric non-negative definite
$2(n+1)\times2(n+1)$    matrix $S\ge \bar X\otimes \bar X$. Then,
equations \eqref{010501-22} describe the evolution of the first two
moments of the solution of \eqref{eq:flip}--\eqref{eq:pbdf})
whose initial distribution is a random vector with the first two
moments 
given by $\bar X$ and   $S$, respectively.

\subsection{The existence and uniqueness of  the  periodic mean and second
moment}
In the first step we show the existence of a periodic solution of 
\eqref{010501-22} that corresponds to the mean and covariance of a
certain probability evolution.
\begin{proposition}
\label{propB1}
There exists a  unique vector $\bar{\bf X}_{\rm per}=(\bar{\bf q}_{\rm per},
    \bar{\bf p}_{\rm per}  ) \in\bbR^{2(n+1)}$  and a non-negative symmetric
    matrix  $C_{\rm per}\ge \bar{\bf X}_{\rm per}\otimes \bar{\bf X}_{\rm per}$
    such that the solution 
   of \eqref{010501-22} with
$$
\Big((\bar{\bf q}(0),
    \bar{\bf p} (0) ) ,C(0)\Big)=\Big((\bar{\bf q}_{\rm per},
    \bar{\bf p}_{\rm per}  ) ,C_{\rm per}\Big)
$$
satisfies
\begin{equation}
\label{030501-22}
\Big((\bar{\bf q}(0),
    \bar{\bf p} (0) ) ,C(0)\Big)=\Big((\bar{\bf q}(\theta),
    \bar{\bf p} (\theta) ) ,C(\theta)\Big).
\end{equation}
In addition, we have 
\begin{equation}
\label{011401-12}
 {C^{(p)}_{x,x}(t)} \ge T_-,\quad x=0,\ldots,n.
\end{equation}
\end{proposition}
The remaining part of this section is devoted to the proof of this results.

\subsubsection{The existence of the periodic first moment}
Let
\begin{equation}
\label{010501-22a}
\left(
  \begin{array}{c}\bar{\bf q}\\
    \bar{\bf p}
    \end{array}\right):=  \int_{-\infty}^0 \cF(s /\theta)
e^{As}{\rm e}_{p,n+1}\dd s.
\end{equation}
Thanks to Proposition \ref{prop012212-21}
the  vector
$(\bar{\bf q},
    \bar{\bf p}  )$ is well defined. One can easily check that the
    solution of the first equation of 
  \eqref{010501-22} starting from the vector is given by 
\begin{equation}
\label{010501-22b}
\bar{\bf X}(t)=\left(
  \begin{array}{c}\bar{\bf q}(t)\\
    \bar{\bf p}(t)
    \end{array}\right):=  \int_{-\infty}^t \cF(s/\theta )
e^{-A (t-s)}{\rm e}_{p,n+1}\dd s.
\end{equation}
and is therefore
$\theta$-periodic. In
  fact, thanks to Proposition \ref{prop012212-21} the periodic
  solution has to be unique. Since the coordinates of
  $\bar{\bf X}(t)$ satisfy the first equation of \eqref{010501-22} we
  conclude that the matrix
$\bar{\bf X}_2(t):=\bar{\bf X}(t)\otimes \bar{\bf X}(t)$ satisfies
$$
\frac{\dd }{\dd t}\bar{\bf X}_2(t)=-A \bar{\bf X}_2(t)- \bar{\bf X}_2(t)
A^T+  \cF(t /\theta) F(t),
$$
and it is given by the formula
\begin{equation}
\label{P2}
\bar{\bf X}_2(t)= 
\int_{-\infty}^t\cF(s/\theta )e^{-A(t-s)} F(s)e^{-A^T(t-s)}\dd s ,\quad t\in\bbR.
\end{equation}

\subsubsection{The existence of the periodic second moment}
 Now we are going to establish the
  existence of a periodic second moment. 
Suppose that $C(t)$ is a periodic solution of the second equation of
\eqref{010501-22}. Using the argument made in the proof of Proposition
\ref{prop010612-21} we can conclude that it satisfies the equation
\begin{equation}
        \label{010612-21a}
        \begin{split}
          C(t) =  \int_{-\infty}^te^{-A(t-s)}\left(\Sigma_2\big({\bf c}_{2}(s)\big)
              + \cF(s/\theta ) F(s)\right) e^{-A^T(t-s)}\dd s\\
          =\int_{-\infty}^te^{-A(t-s)}\Sigma_2\big({\bf c}_{2}(s)\big) e^{-A^T(t-s)}\dd s
          +\bar{\bf X}_2(t)\\
          = \int_{0}^\infty
          e^{-As}\Sigma_2\big({\bf c}_{2}(t-s)\big) e^{-A^T s} \dd s +\bar{\bf X}_2(t)\\
          = \sum_{\ell=0}^\infty \int_{0}^\theta e^{-A(s+\ell \theta)}
          \Sigma_2\big({\bf c}_{2}(t-s)\big) e^{-A^T(s+\ell \theta)} \dd s +\bar{\bf X}_2(t)\\,\quad t\in\bbR,
\end{split}
      \end{equation}
      where the matrix $\Sigma_2$ is defined by \eqref{S2},
      ${\bf c}_{2}(s)$ relates to $C(s)$ via \eqref{c2C} and $F(s)$ is
defined by \eqref{Ft}, using $(\bar{\bf q}(t),
    \bar{\bf p}(t)  )$ instead of $(\bar{\bf q}_\mu(t),
    \bar{\bf p}_\mu(t)  )$. Conversely, any periodic symmetric matrix
    valued function $C(t)$ satisfying \eqref{010612-21a} is a periodic
    solution to the second equation of \eqref{010501-22}.

For $x,x',y=0,\ldots,n$ define
\begin{equation}
\label{gxx}
{g_{x,x',y}(s):=\sum_{\ell=0}^{+\infty}
\Big[e^{-A(s+\ell\theta)}\Big]_{x+n+1,y+n+1}\Big[e^{-A^T(s+\ell\theta)}\Big]_{y+n+1,x'+n+1}.}
 \end{equation}
Consider the following linear mapping: 
${\cal L}:[C(\bbT_\theta)]^{n+1}\to [C(\bbT_\theta)]^{n+1}$, where 
$\bbT_\theta:=\theta\bbT$ is the torus of size $\theta$,  that assigns to a given
vector  of $\theta$-periodic functions
${\bf T}(s)=[T_{0}(s),\ldots,T_n(s)]$ a vector valued function
\begin{equation}
\label{VPT}
{\cal L}{\bf T} :=({\frak G}_0{\bf T},\ldots, {\frak G}_n{\bf T}),
\end{equation}
where 
\begin{align}
\label{021711-21}
 {\frak G}_x{\bf T}(t) = \sum_{y=0}^n \int_0^{\theta}{{\frak G}_{x,y}(s)} T_y\left(t-s\right)\dd s.
\end{align}
 Here     
\begin{equation}
\label{011504-22}
{{\frak G}_{x,y}(s)} = 4 \ga  g_{x,x,y}(s),\quad y=0,\ldots,n.
\end{equation} {Obviously, from \eqref{gxx}, we have ${\frak G}_{x,y}(s)\ge0$.}
Note also that although $ g_{x,x',y}(\cdot)$ need not be $\theta$-periodic
the functions ${\frak G}_x{\bf T}(t)$, $x=0,\ldots,n$  are
$\theta$-periodic. In addition, if $C(t)$ satisfies
\eqref{010612-21a}, then
\begin{equation}
\label{P2c}
{\bf c}_{2}(t) ={\cal L}{\cal T}({\bf c}_{2})(t) +\bar{\bf p}^2(t),
\end{equation}
where for a given ${\bf T}^T=(T_0,\ldots,T_n)\in\bbR^{n+1}$
$$
{\cal T} ({\bf T})= \left(
  \begin{array}{c}
T_-\\
  T_1 \\
    \vdots\\
     T_n
  \end{array}
    \right) ,\quad \bar{\bf p}^2(t) = \left(
  \begin{array}{c}
\bar p_0^2(t)\\
    \vdots\\
    \bar p_n^2(t)
  \end{array}
    \right).
$$
Conversely, by finding a solution ${\bf c}_{2}$ of \eqref{P2c} one can define then a
$\theta$-periodic function $C(t)$ by the right hand side of
\eqref{010612-21a}. The entries of the function corresponding to
{$C_{x,x}$, $x=n+1,\ldots,2n+1$} coincide with the coordinates of the
vector ${\bf c}_{2}$, by virtue of \eqref{P2c}. Thus, the function $C(t)$
solves equation \eqref{010612-21a}. We have reduced therefore the
problem of finding a periodic solution to the second equation of
\eqref{010501-22} to solving equation \eqref{P2c}.

\subsubsection{Solution of (\ref{P2c})}

Let 
{$$
{\cal C}_+:=[{\bf T}=(T_0,T_1,\ldots,T_n):\,T_x\in
C(\bbT_\theta)\,\mbox{and }T_x\ge T_-,\,x=0,\ldots,n].
$$
It is a closed subset of $\Big(C(\bbT_\theta)\Big)^{n+1}$, equipped
with the norm 
$$
|\!\|{\bf T}\|\!|:=\max\{ \|T_x\|_\infty,\,x=0,\ldots,n\}.
$$}
 Consider the mapping
\begin{equation}
\label{fP1}
{\frak T}=({\frak T}_0,\ldots, {\frak T}_n): {\cal C}_+\to \Big(C(\bbT_\theta)\Big)^{n+1},
\quad\mbox{where }\quad {\frak T}{\bf T}:={\cal L}{\cal T}({\bf T})+\bar{\bf p}^2.
\end{equation}


{
  Using the notation of \eqref{021711-21} and \eqref{P2c} we have
$$
{\frak T}_x({\bf T})(t):=T_-\int_0^{\theta} {\frak G}_{x,0}(s)\dd
s+\sum_{x'=1}^n\int_0^\theta
{\frak G}_{x,x'}(s)T_{x'}\left(t-s\right)\dd s+\bar
p_x^2(t),\quad x=0,\ldots,n.
$$ 
{Comparing \eqref{010612-21aa} with \eqref{010612-21a}, after time averaging over a period, it is easy to identify}
  \begin{equation}
    \label{eq:51}
    \int_0^{\theta} {\frak G}_{x,y}(s)\dd s = M_{x,y}
  \end{equation}
  defined by \eqref{eq:54}. The matrix $[M_{x,y}]_{x,y=0}^n$ is symmetric,
  bi-stochastic  
  {(as can be easily seen from \eqref{eq:54}).}.
  It also follows immediately that
  \begin{equation}
  \label{022910-21-b}
  1= \sum_{y=0}^n
  \int_0^{\theta} {\frak G}_{x,y}(s)\dd s,\quad x=0,1,\ldots,n
\end{equation}
and, as a consequence, ${\frak T}\big({\cal C}_+\big)\subset {\cal
  C}_+$. Furthermore, we claim that $M_{x,y}>0$
for all $x,y=0,\ldots,n$. Indeed,
a simple
calculation, using \eqref{A} and \eqref{012212-21}, yields
\begin{align}
\label{022212-21}
   &\big[ (\la+A)^{-1}\big]_{x+n+1,y+n+1}=
    \sum_{j=0}^n\frac{\la \psi_j(x)\psi_j(y) }{\la^2+2\ga \la+\mu_j},\notag\\
  &\big[ (\la+A)^{-1}\big]_{x,y+n+1}= \sum_{j=0}^n\frac{
     \psi_j(x)\psi_j(y)}{\la^2+2\ga \la+\mu_j}, \\
    &\big[ (\la+A)^{-1}\big]_{x+n+1,y}
    =-\sum_{j=0}^n\frac{  \mu_{j}\psi_j(x)\psi_j(y) }{\la^2+2\ga \la+\mu_j}, \notag\\
  &  \big[ (\la+A)^{-1}\big]_{x,y+n+1}= \sum_{j=0}^n\frac{ \psi_j(x)\psi_j(y)}{\la^2+2\ga \la+\mu_j}. \notag
\end{align}
The poles of the meromorphic functions appearing in \eqref{022212-21}
are given by 
{\begin{equation}
\label{lapm}
\la_{j,\pm}=-\Big(\ga\pm\sqrt{\ga^2-\mu_j}\Big).
\end{equation}}
Suppose that $M_{x,y}=0$ for some
  $x,y$. From \eqref{eq:51} we conclude then that
  \begin{equation}
    \label{eq:51a}
  0=   M_{x,y}=\int_0^{\theta} {\frak G}_{x,y}(s)\dd s
  =4\ga\int_0^{+\infty}\Big[e^{-As}\Big]^2_{x+n+1,y+n+1}\dd s,
\end{equation}
which in turn would implies that
$\Big[e^{-As}\Big]_{x+n+1,y+n+1}\equiv0$ for all $s\ge0$, thus also
$$
0\equiv \big[ (\la+A)^{-1}\big]_{x+n+1,y+n+1}=
    \sum_{j=0}^n\frac{\la \psi_j(x)\psi_j(y) }{\la^2+2\ga \la+\mu_j}.
$$
As a result, we conclude that $ \psi_j(x)\psi_j(y)=0$, for all
$j=0,\ldots,n$, which is impossible.
}

\bigskip

We shall show that the mapping ${\frak T}$ has  a unique fixed point
in ${\cal C}_+$ by proving  that the mapping
  is a  contraction   in the  
  norm  $|\!\|\cdot\|\!|$.
Indeed, for ${\bf T}_j^T=[T_{j,0},T_{j,1},\ldots,T_{j,n}]$, $j=1,2$, we
have
 \begin{align*}
&
|{\frak T}_x({\bf T}_1)(t) -{\frak T}_x({\bf T}_2)(t)|=\Big|
  \sum_{x'=1}^n\int_0^\theta {\frak G}_{x,x'}(s)\Big[T_{1,x'}\left(t-s\right)-
  T_{2,x'}\left(t-s\right)\Big]\dd s\Big|\\
&
\le 
  \sum_{x'=1}^n\int_0^\theta {\frak G}_{x,x'}(s) \Big| T_{1,x'}\left(t-s\right)-
  T_{2,x'}\left(t-s\right)\Big|\dd s\\
&
\le \sum_{x'=1}^n \left(\int_0^\theta {\frak G}_{x,x'}(s)\dd s\right)\|
  T_{1,x'} -
  T_{2,x'}\|_{\infty}\\
&
\le \left(1- \int_0^\theta {\frak G}_{x,0}(s)\dd s\right)|\!\|{\bf T}_1-{\bf T}_2|\!\|,\quad x=0,\ldots,n.
\end{align*}
Therefore
 \begin{align*}
|\!\|{\frak T}({\bf T}_1)-{\frak T}({\bf T}_2)\|\!|
\le \rho |\!\|{\bf T}_1-{\bf T}_2|\!\|,
\end{align*}
where
$$
\rho:=\max \Big[ 1- \int_0^\theta {\frak G}_{x,0}(s)\dd s ,\,x=0,\ldots,n\Big]<1.
$$
We have shown that
$
\|{\frak T}({\bf T}_1) -{\frak T} ({\bf T}_2)\|_\infty
\le \rho\|{\bf T}_1-{\bf T}_2\|_{\infty}$
and the existence  of a unique fixed point follows.
This ends the proof of Proposition \ref{propB1}.
\qed

\subsection{The end of the proof of Theorem \ref{periodic}}

Suppose now that $\nu$ is a probability law whose first and second
moments are $\theta$-periodic, e.g. it could be a Gaussian
distribution with the mean and the second moment given by ${\bf
  P}_{\rm per}$ and $C_{\rm per}$, respectively.
Denote by 
$$
{\cal P}_{s,t}F({\bf q}, {\bf p})=\int_{\bbR^{2(n+1)}}F({\bf q}', {\bf
  p}'){\cal P}_{s,t}({\bf q}, {\bf p}; \dd {\bf q}', \dd {\bf p}')
$$ the evolution family of transition
probability operators corresponding to the dynamics described by
\eqref{eq:flip}--\eqref{eq:pbdf}. 
Consider the event
$
E:=[N_x({\theta})=0,\,x=1,\ldots,n].
$ 
We have $\bbP[E]>0$. Suppose that the dynamics  starts at  $({\bf q},
{\bf p})$. Then, for any $F\ge0$ we can write
\begin{equation}
\label{030801-22}
{\cal P}_{0,\theta}F({\bf q}, {\bf p})\ge \bbE\Big[F({\bf q}(\theta), {\bf p}(\theta),\,E\Big]=\bbP[E]{\cal Q}_{0, \theta}F({\bf q}, {\bf p}),
\end{equation}
{where ${\cal Q}_{s,t}$ is the evolution family of transition probability operators for the
non-homogeneous Ornstein-Uhlenbeck process $V(t)$ that corresponds to the
generator 
$  \mathcal G_t^{(g)} =  \mathcal A_t
 + 2  \gamma S_{-} $, see \eqref{eq:7} and \eqref{eq:8}. Using
 the hypoellipticity  of the
time homogeneous Ornstein-Uhlenbeck process $U(t)$ that corresponds to the
generator 
 $
 \mathcal G^{(g)}:=\mathcal A
 + 2  \gamma S_{-} ,
 $
 where
 $$
 \mathcal A = \sum_{x=0}^n p_x \partial_{q_x}
  + \sum_{x=0}^n  (\Delta_N q_{x}-\om^2_0q_x) \partial_{p_x},
$$
  see Section \ref{appA3}
 below, we can prove that there
 exist strictly positive transition
 probability density kernels $\rho_{s,t}$ corresponding to ${\cal
   Q}_{s,t}$. Suppose that the time homogeneous Ornstein-Uhlenbeck  process $U(t;u)$
 satisfies 
   the S.D.E.
 $$
dU(t;u)=-A_EU(t;u)\dd t+\sqrt{4 \gamma T_-}\Sigma \dd W(t),\quad U(0;u)=u
$$
where $A_E$ and $\Sigma$   are given by   $2\times
2$ block matrices, whose entries are $(n+1)\times (n+1)$ matrices 
\begin{equation}
\label{Sig}
A_E=
\left(
  \begin{array}{cc}
    0&-{\rm Id}_{n+1}\\
    -\Delta_{\rm N} +\om_0^2& 2\ga E
  \end{array}
\right)\quad \mbox{and}\quad \Sigma =
\left(
  \begin{array}{cc}
    0&0\\
    0& E
  \end{array}
\right),
\end{equation}
with
\begin{equation}
  \begin{split}
E: =
       \begin{bmatrix}
1& 0 &  \dots&0\\
                     0& 0   &\dots&0\\
                     \\
                     \vdots   & \vdots & \vdots&\vdots\\
 0& 0   & \dots &0
                        \end{bmatrix}.
 \end{split}
\label{eq:22}
\end{equation}
Here
$$
\dd W(s) =\left(\begin{array}{c}
                  dw_0^{(q)}(s)\\
                  \vdots\\
                  dw_n^{(q)}(s)\\
                  dw_0^{(p)}(s)\\
                                  \vdots\\
                  dw_n^{(p)}(s)
\end{array}\right).
$$
is a $2(n+1)$-dimensional standard Wiener process. Due to the
hypoellipticity, the  probability distribution of $U(t;u)$ have
densities that are given by $C^\infty$
smooth Gaussians.}

{The
non-homogeneous Ornstein-Uhlenbeck process $V(t;v)$ that corresponds to the
$  \mathcal G_t^{(g)}$ and  satisfies $V(0;v)=v$, can be described by
(cf \eqref{epq})
the solution of 
$$
dV(t;v)=\Big[-A_EV(t;v)+{\cal F}(t){\rm e}_{p,n+1}
\Big]\dd t+\Sigma \dd W(t),\quad V(0;v)=v.
$$
Hence
\begin{align*}
&V(t;v)=e^{-A_Et}v+\int_0^te^{-A_E(t-s)}\Sigma \dd
                 W(s)+\int_0^te^{-A_E(t-s)}{\cal F}(s){\rm e}_{p,n+1}\dd s\\
  &
    =U(t;v)+\int_0^te^{-A_E(t-s)}{\cal F}(s){\rm e}_{p,n+1}\dd s.
\end{align*}
Thus the distribution of $V(t;v)$ has
a density that is also given by a $C^\infty$
smooth Gaussian.}

 Thanks to \eqref{030801-22} we conclude that
\begin{equation}
\label{040801-22}
{\cal P}_{0, \theta}({\bf q}, {\bf p}; \dd {\bf q}', \dd {\bf p}')
\ge c_*\rho_{0,\theta}({\bf q}, {\bf p};  {\bf q}',  {\bf p}'), {\bf q}' \dd {\bf p}',
\end{equation}
where $c_*:=\bbP[E]$.
Then, $\nu_{0,t}:=\nu {\cal P}_{0,t}$ describes the law of $({\bf q}(t),
    {\bf p}(t)  )$ with the prescribed initial data. Thanks to
    Proposition \ref{propB1} we can see that {the total energy
      ${\cal H}(t):=\sum_{x=0}^n{\cal E}_x(t)$} (see \eqref{Ex}) is a Lyapunov
    function for the above system, since $\bbE{\cal H}(t)$ is $\theta$-periodic. 
The above implies that the family of laws $\{\nu_{0,t},\,t\ge0\}$ is
tight in $\bbR^{2(n+1)}$. Thus, also the family
$\mu_N:=N^{-1}\int_0^{N\theta}\nu_{0,s}\dd s$ is tight. Suppose that $\mu_\infty$
is its limiting measure, i.e. there exists a sequence $N'\to+\infty$
such that $\mu_{N'}\to\mu_\infty$, in the topology of  weak
convergence. Since  ${\cal
  P}_{s,t}$ has the Feller property one can easily conclude that $\mu_\infty {\cal
  P}_{0,\theta}=\mu_\infty$. Hence $\mu_s^P:=\mu_\infty {\cal
  P}_{0,s}$ , $s\in[0,+\infty)$ is a periodic stationary state. 

  Suppose that $\mu(\dd {\bf q}, \dd {\bf p})=f ({\bf
  q}, {\bf p}) \dd {\bf q} \dd {\bf p}$, where  $f$ is a $C^\infty$ smooth
probability density. One can show, using the regularity theory of
stochastic differential equations, that  $\mu {\cal P}_{0,\theta}$ is absolutely
continuous w.r.t. the Lebesgue measure and its density is also
$C^\infty$  smooth, see e.g. \cite[Corollary III.3.4, p. 303]{GS79}. This allows us to conclude further that $\mu {\cal P}_{0,\theta}$ is absolutely continuous, provided that $\mu$ is absolutely
continuous. We shall  denote by ${\cal P}_{0,\theta}$ the corresponding operator induced 
on $L^1(\bbR^{2(n+1)})$.  The operator ${\cal Q}_{0,\theta}$
corresponding to the Gaussian dynamics transforms  $\mu_\infty$ into an
absolutely continuous measure.  Thanks to \eqref{040801-22} we conclude that
\begin{equation}
\label{040801-22a}
\mu_\infty (\dd {\bf q}', \dd {\bf p}')=\mu_\infty{\cal P}_{0,\theta}(\dd {\bf q}', \dd {\bf p}')\ge c_*\mu_\infty{\cal Q}_{0,\theta}( {\bf q}', {\bf p}')\dd {\bf q}' \dd {\bf p}'.
\end{equation}
Therefore the   singular part of $\mu_\infty $ is of at most mass
$1-c_*$. Since ${\cal P}_{0,\theta}$ transforms the space of abolutely continuous
measures into itself, both the 
singular and absolutely continuous parts of $\mu_\infty $, after
normalization, become invariant   under ${\cal P}_{0,\theta}$. Iterating
this procedure we conclude, after $m$ steps, that the singular part
can be of at most mass $(1-c_*)^m$, which eventually leads to the
conclusion that the measure $\mu_\infty$ is absolutely continuous. The
respective density is positive, due to \eqref{040801-22}. This ends
the proof of Theorem \ref{periodic}.\qed

\subsection{Hypoellipticity of $  \mathcal G^{(g)}$}

\label{appA3}

We show that the operator
$
 \Big(\mathcal G^{(g)}\Big)^*-\partial_t
$
is hypoelliptic in $\bbR\times\bbR^{2(n+1)}$. 
Here
\begin{equation}
  \label{eq:7-1b}
 \Big(\mathcal G^{(g)}\Big)^*-\partial_t= {\cal X}_0+{\cal X}_1^2 +2  \ga,
\end{equation}
where
$$
{\cal X}_0=-\partial_t-\mathcal A +  2 \ga p_0\partial_{p_0}  \quad
\mbox{and}\quad {\cal X}_1=\sqrt{ 2 \ga T_-}\partial_{z_0} .
$$
Let 
$
{\cal X}_1^{(0)}:={\cal X}_1=\sqrt{ 2\ga T_-}\partial_{p_0} .
$
Since
$$
\big[{\cal X}_0, \partial_{p_0}\big]   
= \partial_{q_0}-2\gamma\partial_{p_0}
$$
the commutator
\begin{align*}
  {\cal X}_2^{(0)}:=\big[{\cal X}_0, {\cal X}_1^{(0)}\big]=\sqrt{ 2\ga T_-}
    \big[{\cal X}_0, \partial_{p_0}\big]   
    =\sqrt{ 2\ga T_-}\{\partial_{q_0}-2\gamma\partial_{p_0}\}.
\end{align*}

Next, we have
$$
\big[{\cal X}_0,\partial_{q_0}\Big]=-(1+\om_0^2)\partial_{p_0}+\partial_{p_1}.
$$
Hence,
\begin{align*}
  &{\cal X}_1^{(1)}:=\big[{\cal X}_0, {\cal X}_2^{(0)}\big]
    =\sqrt{ 2\ga T_-}\Big\{\big[{\cal
    X}_0,\partial_{q_0}]-2\gamma\big[{\cal
    X}_0,\partial_{p_0}]\Big\}\\
  &
   =\sqrt{ 2\ga T_-}\partial_{p_1}+c_0^{(0)} \partial_{p_0}+d_0^{(0)}  \partial_{q_0}
\end{align*}
for some constants $c_0^{(0)}$ and $d_0^{(0)}$.
Note that,
$
\big[{\cal X}_0,\partial_{p_1}\Big]=\partial_{q_1}.
$
Therefore,
\begin{align*}
  &{\cal X}_2^{(1)}:=\big[{\cal X}_0, {\cal X}_1^{(1)}\big]
   =\sqrt{ 2 \ga T_-}\big[{\cal X}_0,
    \partial_{p_1}\big]+c_0^{(0)}\big[{\cal X}_0,
    \partial_{z_0}\big]+d_0^{(0)} \big[{\cal X}_0, \partial_{y_0}\big]\\
  &
    =\sqrt{ 2\ga T_-}
    \partial_{q_1}+c_1^{(1)} \partial_{p_1}+c_0^{(1)} \partial_{p_0}+d_0^{(1)} \partial_{q_0}.
\end{align*}
for some constants $c_0^{(1)}$,  $c_1^{(1)}$  and
$d_0^{(1)}$.
Continuing  calculations along those lines we get that
\begin{align*}
  &{\cal X}_1^{(m)}:=\big[{\cal X}_0, {\cal X}_2^{(m-1)}\big]
    =\sqrt{ 2 \ga T_-}
    \partial_{p_m}+\sum_{j=0}^{m-1}c_j^{(m)}
    \partial_{p_j}+\sum_{j=0}^{m-1}d_j^{(m)} \partial_{q_j},\\
  &{\cal X}_2^{(m)}:=\big[{\cal X}_0, {\cal X}_2^{(m-1)}\big]
    =\sqrt{ 2\ga T_-}
    \partial_{q_m}+c_m^{(m)}
    \partial_{p_m}+\sum_{j=0}^{m-1}c_j^{(m)}(t)
    \partial_{z_j}+\sum_{j=0}^{m-1}d_j^{(m)} \partial_{q_j}.
\end{align*}
for some constants $c_j^{(m)},d_j^{(m)}$.  In conclusion, we can see that
${\cal X}_0$, ${\cal X}_1^{(m)}$ and ${\cal X}_2^{(m)}$, $0\le m\le n$ generate the
tangent space to $\bbR\times \bbR^{2(n+1)}$ so the operator $
 \Big(\mathcal G^{(g)}\Big)^*-\partial_t
$
is hypoelliptic in $\bbR\times\bbR^{2(n+1)}$ by virtue of the
H\"ormander theorem.

\section{Green Functions convergence}
\label{sec:green-fnc}

Recall that $G_{\om_0}$ and $G^n_{\om_0}$ are the Green's functions
corresponding to
$\om_0^2-\Delta$ and $\om_0^2-\Delta_{\rm N}$, where $\Delta$ is the
free lattice laplacian on $\mathbb Z$ and $\Delta_{\rm N}$ is the Neumann
discrete laplacian on $\{0,1,\dots , n\}$, 
see Sections \ref{sec2.6} and \ref{sec2.7}, respectively.

\subsection{Estimates on oscillating sums}
\label{sec:estim-oscill-sums}

 Define $\chi_n(x)$ as the $n+1$-periodic extension of 
$\chi_n(x):=(1+x)\wedge (n+2-x)$, $x\in[0,n+1]$.
Suppose that $\Phi:\bbR^2\to \mathbb C$ is a $\theta,\theta'$-periodic function in
each variable respectively.
Denote
\begin{align*}
H_{x,x'}^{(n)}
  =\frac{1}{(n+1)^2}\sum_{j,j'=0}^n\Phi\left(\frac{\theta
  j}{n+1}, \frac{\theta' j'}{n+1}\right) \exp\left\{\frac{ 2i\pi  
  j
  x}{n+1}\right\} \exp\left\{\frac{ 2i\pi   j'x'}{n+1}\right\}
\end{align*}
 for $x,x'\in\bbZ$.
\begin{lemma}
\label{lm012812-21}
Suppose that $\Phi$ is of $C^k$-class for some $k\ge1$. Then, there
exists $C$ such that 
\begin{equation}
\label{012812-21}
\Big|H_{x,x'}^{(n)}\Big|\le\frac{C}{\chi^\ell_n(x) \chi^{\ell'}_n(x')},\quad
x,x'\in \bbZ,\,n\ge 1,\,\ell,\ell'\ge0,\,\ell+\ell'\le k.
\end{equation}
\end{lemma}
\proof To simplify the notation we suppose that $\theta=\theta'=1$.
Summation by parts yields
\begin{align*}
&H_{x,x'}^{(n)}
=\frac{1}{(n+1)^2}\sum_{j,j'=0}^n\Phi\left(\frac{
  j}{n+1}, \frac{ j'}{n+1}\right)
\frac{\exp\left\{\frac{ 2 i\pi
  x}{n+1}\right\}-1}{\exp\left\{\frac{ 2 i\pi
  x}{n+1}\right\}-1}\exp\left\{\frac{ 2 i\pi
  jx}{n+1}\right\} \exp\left\{\frac{  2i\pi j'x'}{n+1}\right\} 
\end{align*}
\begin{align*}
&
  =\frac{1}{(n+1)^2 [\exp\left\{\frac{  i\pi
  x}{n+1}\right\}-1]} \sum_{j,j'=0}^n\Big[\Phi\left(\frac{
  j-1}{n+1}, \frac{ j'}{n+1}\right)-
  \Phi\left(\frac{
  j}{n+1}, \frac{ j'}{n+1}\right)\Big]
\\
&
\times \exp\left\{\frac{ 2 i\pi
  jx}{n+1}\right\} \exp\left\{\frac{ 2 i\pi j'x}{n+1}\right\}
\end{align*}
Since $\Phi$ is of $C^1$ class
$$
\Big|\Phi\left(\frac{
  j-1}{n+1}, \frac{ j'}{n+1}\right)-
  \Phi\left(\frac{
  j}{n+1}, \frac{ j'}{n+1}\right)\Big|\le \frac{C}{n+1}
$$
for some constant $C>0$. 
In addition, there exists $c>0$ such that
$$
\left|\exp\left\{\frac{ 2 i\pi
  x}{n+1}\right\}-1\right|\ge \frac{c\chi_n(x)  }{n+1}
$$
for
$x,x'\in \bbZ,\,n\ge1$.
Thus, there exists $C>0$ such that
$$
|H_{x,x'}^{(n)}|\le \frac{C}{\chi_n(x)},\quad x,x'\in\bbZ. 
$$
Iterating this argument in the regularity degree $k$ of $\Phi$ we conclude \eqref{012812-21}.\qed

\subsection{Application}

An application concerns the approximation of the Green's function
$G_{\om_0}$ by $G_{\om_0}^n$ along the diagonal.
  \begin{lemma}\label{lem-ful}
   We have
\begin{equation}
    \label{010703-22}
G_{\om_0}^n(y,y)=G_{\om_0}(0)+\tilde H^{(n)}(y) +O\Big(\frac{1}{n}\Big),\quad y=0,\ldots,n,\,n\ge1.
\end{equation}
 Here for some constant $C>0$ we have
\begin{align}
\label{042912-21}
 |\tilde H^{(n)}(y)|\le \frac{C}{\chi^2(y)} ,\quad y=0,\ldots,n,\,n\ge1.
\end{align}
  \end{lemma}
\proof
Using the definition of the Green's function \eqref{Gpsi} (with $\ell=0$) and formulas
\eqref{laps}
  we obtain
\begin{align*}
  G_{\om_0}^n(y,y)
=\frac{ 1}{n+1}\sum_{j=0}^n \Xi\left(\frac{  j}{n+1}\right) \left[1+\cos\left(\frac{\pi j(2y+1)}{n+1}\right) \right]+O\Big(\frac{1}{n}\Big),
\end{align*}
where 
$$
\Xi\left(u \right)=\left\{4\sin^2\left(\frac{\pi
     u}{2}\right)+\om_0^2  \right\}^{-1}.
$$
As a result we write $G_{\om_0}^n(y,y)$ in the form
\eqref{010703-22}, with 
\begin{align*}
\tilde H_y^{(n)}=\frac{ 1}{2(n+1)}\sum_{j=-n-1}^n\cos\left(\frac{\pi
  j(2y+1)}{n+1}\right) \Xi\left(\frac{  j}{n+1}\right).
\end{align*} Estimate \eqref{042912-21} is then a consequence of Lemma \ref{lm012812-21}.
\qed

  \section{Proofs of Lemmas \ref{lm011504-22}, \ref{lm021504-22} and \ref{lm031504-22}}

\label{appC}

\subsection{Proof of Lemma \ref{lm011504-22}}

For $m\in\bbZ$ and $g=(g_0,\ldots,g_n)\in\bbR^{n+1}$ define 
$$
S(m):=\int_0^{+\infty}e^{-2\pi i
  ms/\theta}e^{-As}\Sigma(g) e^{-A^Ts} \dd s=
\left[
  \begin{array}{cc}
  S^{q}(m)&   S^{qp}(m)\\
&\\
     S^{pq}(m)&   S^{p}(m)
  \end{array}
\right],
$$
where $\Sigma_2$ is defined in \eqref{S2}.
Note that
\begin{equation}
\label{MmD}
\sum_{y=0}^nM_{x,y}(m)g_y=S^{p}_{x,x}(m).
\end{equation}
Arguing similarly as in the proof of \eqref{SA} we get 
\begin{equation}
  \label{DA}
A S(m)+ S(m) A^T +\frac{2\pi i
  m}{\theta}S(m)=\Sigma_2(g).
\end{equation}
Denote
$$
 \tilde{ S}^{q,p}_{j,j'}(m)=\sum_{x,x'=0}^nS^{q,p}_{x,x'}(m)\psi_j(x) \psi_{j'}(x')
$$ 
Following the same manipulations as those leading to \eqref{163011-21a} 
we obtain
\begin{align}
\label{011304-22a}
 &\tilde S^{p}_{j,j'} (m)=\frac12 \tilde S^{q}_{j,j'}(m)
   \Big[\mu_{j'}+\mu_j + \Big(2\ga  +\frac{2\pi i
  m}{\theta}\Big) \frac{\pi i
  m}{\theta} \Big]
   , \notag\\
  &\tilde S^{q}_{j,j'}(m)\Big[\mu_{j'}-\mu_j - \frac{2\pi i
  m}{\theta}\Big(2\ga  +\frac{2\pi i
  m}{\theta}\Big)\Big]+2\Big(2\ga  +\frac{2\pi i
  m}{\theta}\Big)
\tilde S^{qp}_{j,j'} (m) =0,\\
  &  \tilde S^{qp}_{j,j'}(m)\Big(\mu_{j}-\mu_{j'} \Big)+ \frac{2\pi i
  m}{\theta}\tilde S^{q}_{j,j'}(m)\mu_{j'}  +\Big(4\ga  +\frac{2\pi i
  m}{\theta}\Big)\tilde S^{p}_{j,j'}(m)= \tilde F_{j,j'}
   \notag
\end{align}
and
$
\tilde F_{j,j'}=\sum_{x=0}^n\psi_j(x)\psi_{j'}(x)g_x.
$
{Determine $\tilde S^{qp}_{j,j'}$ from the second equation. It equals
\begin{align*}
\tilde S^{qp}_{j,j'}=[2\Big(2\ga  +\frac{2\pi i
  m}{\theta}\Big)]^{-1}\tilde S^{q}_{j,j'}\Big[\mu_{j}-\mu_{j'} + \frac{2\pi i
  m}{\theta}\Big(2\ga  +\frac{2\pi i
  m}{\theta}\Big)\Big] .
\end{align*}
Substituting into the third equation we get
\begin{align*}
&  \tilde S^{p}_{j,j'}= \Big(4\ga  +\frac{2\pi i
  m}{\theta}\Big)^{-1}\Big\{ F_{j,j'}-\tilde S^{qp}_{j,j'}\Big(\mu_{j}-\mu_{j'}  \Big) - \frac{2\pi i
  m}{\theta}\tilde S^{q}_{j,j'}\mu_{j'}  \Big\}
   \\
&
=\Big(4\ga  +\frac{2\pi i
  m}{\theta}\Big)^{-1}\Big\{ F_{j,j'}
 -\tilde S^{q}_{j,j'}\Big\{[2\Big(2\ga  +\frac{2\pi i
  m}{\theta}\Big)]^{-1}(\mu_{j}-\mu_{j'})^2 + \frac{\pi i
  m}{\theta} (\mu_{j}+\mu_{j'}) \Big\}
\end{align*}
Hence
\begin{align}
  \label{020505-22}
&  \tilde S^{q}_{j,j'}
=\Big\{[2\Big(2\ga  +\frac{2\pi i
  m}{\theta}\Big)]^{-1}(\mu_{j}-\mu_{j'})^2 + \frac{\pi i
  m}{\theta} (\mu_{j}+\mu_{j'}) \Big\}^{-1} F_{j,j'}
\\
&
-\tilde S^{p}_{j,j'}\Big(4\ga  +\frac{2\pi i
  m}{\theta}\Big) \Big\{[2\Big(2\ga  +\frac{2\pi i
  m}{\theta}\Big)]^{-1}(\mu_{j}-\mu_{j'})^2 + \frac{\pi i
  m}{\theta} (\mu_{j}+\mu_{j'}) \Big\}^{-1}.\notag
\end{align}
On the other hand, from the first equation of \eqref{011304-22a} we
conclude that
\begin{align}
                \label{030505-22}
  \tilde S^{q}_{j,j'} =2\tilde S^{p}_{j,j'}\Big[\mu_{j'}+\mu_j + \Big(2\ga  +\frac{2\pi i
  m}{\theta}\Big) \frac{\pi i
  m}{\theta} \Big]^{-1}.
\end{align}
Therefore, comparing \eqref{020505-22} with \eqref{030505-22}, we get
\begin{align*}
&  2\tilde S^{p}_{j,j'}\Big[\mu_{j'}+\mu_j + \Big(2\ga  +\frac{2\pi i
  m}{\theta}\Big) \frac{\pi i
  m}{\theta} \Big]^{-1} \\
&
=\Big\{[2\Big(2\ga  +\frac{2\pi i
  m}{\theta}\Big)]^{-1}(\mu_{j}-\mu_{j'})^2 + \frac{\pi i
  m}{\theta} (\mu_{j}+\mu_{j'}) \Big\}^{-1} F_{j,j'}
\\
&
-\tilde S^{p}_{j,j'}\Big(4\ga  +\frac{2\pi i
  m}{\theta}\Big) \Big\{[2\Big(2\ga  +\frac{2\pi i
  m}{\theta}\Big)]^{-1}(\mu_{j}-\mu_{j'})^2 + \frac{\pi i
  m}{\theta} (\mu_{j}+\mu_{j'}) \Big\}^{-1}
\end{align*}and
\begin{align*}
&  2\tilde S^{p}_{j,j'}\Big[\mu_{j'}+\mu_j + \Big(2\ga  +\frac{2\pi i
  m}{\theta}\Big) \frac{\pi i
  m}{\theta} \Big]^{-1} \\
&
\times \Big\{[2\Big(2\ga  +\frac{2\pi i
  m}{\theta}\Big)]^{-1}(\mu_{j}-\mu_{j'})^2 + \frac{\pi i
  m}{\theta} (\mu_{j}+\mu_{j'}) \Big\}
+\tilde S^{p}_{j,j'}\Big(4\ga  +\frac{2\pi i
  m}{\theta}\Big)= F_{j,j'}
\end{align*}
This leads to
$$
\tilde S^{p}_{j,j'}=\sum_{y}\Theta_m(\mu_j,\mu_{j'})\psi_j(y)g_y\psi_{j'}(y),
$$
with
\begin{align}
  \label{Thetam}
& \Theta_m(c,c') :=4\ga\Big(4\ga  +\frac{2\pi i
  m}{\theta}\Big)^{-1} \Bigg\{2 \Big(4\ga  +\frac{2\pi i
  m}{\theta}\Big)^{-1}\Big[c+c' + \Big(2\ga  +\frac{2\pi i
  m}{\theta}\Big) \frac{\pi i
  m}{\theta} \Big]^{-1} \notag\\
&
\times \Big\{[2\Big(2\ga  +\frac{2\pi i
  m}{\theta}\Big)]^{-1}(c-c')^2 + \frac{\pi i
  m}{\theta} (c+c') \Big\}
+1 \Bigg\}^{-1}
\end{align}
From \eqref{MmD} we conclude that
\begin{align}
\label{011404-22}
M_{x,y}(m)= \sum_{j,j'=0}^n\Theta_m(\mu_j,\mu_{j'}) \psi_j(x)\psi_{j'}(x) \psi_j(y)\psi_{j'}(y)
\end{align}
Note that $\Theta_0(c,c')=\Theta(c,c')$ defined in  \eqref{eq:47}.}
As in \eqref{071504-22}, for any sequence $(f_x)\in\mathbb C^{n+1}$ we can write
\begin{align*}
\sum_{x,y=0}^n(\delta_{x,y}-M_{x,y}(m))f_y^\star f_x = \sum_{j,j'=0}^n \left(1- \Theta_m(\mu_j,\mu_{j'})\right)
     \left|\sum_{x=0}^n \psi_j(x)f_x\psi_{j'}(x)\right|^2.
\end{align*}
We have
\begin{align}
  \label{070505-22}
  &1- \Theta_m(c,c')\\
  &
    = \Bigg\{\frac{\pi i
  m}{\theta} +\Big[c+c' +2 \Big(\ga  +\frac{\pi i
  m}{\theta}\Big) \frac{\pi i
  m}{\theta}\Big]^{-1} \Big[\frac14\Big(\ga  +\frac{\pi i
  m}{\theta}\Big)^{-1}(c-c')^2 + \frac{\pi i
  m}{\theta} (c+c') \Big]\Bigg\}\notag\\
&
\times \Bigg\{\Big[ c+c' +2 \Big(\ga  +\frac{\pi i
  m}{\theta}\Big) \frac{\pi i
  m}{\theta} \Big]^{-1} \Big[\frac14\Big(\ga  +\frac{\pi i
  m}{\theta}\Big)^{-1}(c-c')^2 + \frac{\pi i
  m}{\theta} (c+c') \Big]
+\Big(2\ga  +\frac{\pi i
  m}{\theta}\Big) \Bigg\}^{-1}.\notag
\end{align}
{It is easy to see from \eqref{Thetam} that
  \begin{equation}
    \label{050505-22}
\lim_{m\to+\infty}\Big(1- \Theta_m(c,c')\Big)=1
\end{equation}
uniformly in $  c,c'\in[0,\om_0^2+4]$.}

{To prove  \eqref{021504-22} it suffices to show that
  there exists $ {\frak C}_*>0$
such that
\begin{align}
\label{040505-22}
|1- \Theta_m(c,c')|\ge  {\frak C}_*,\quad |m|\ge1,\,c,c'\in[0,\om_0^2+4].
\end{align}
In light of \eqref{050505-22} to show \eqref{040505-22} it suffices therefore to
prove that
\begin{align}
\label{040505-22a}
|1- \Theta_m(c,c')|\not=0,\quad |m|\ge1,\,c,c'\in[0,\om_0^2+4].
\end{align}
Suppose that $m\not=0$ and  $1- \Theta_m(c,c')=0$ for some $c,c'\in[0,\om_0^2+4]$.
Then, \eqref{070505-22} implies that
\begin{align*}
0=2\al i \Big[c+c' +(\ga  +\al i) \al i\Big]+
  \frac{(c-c')^2 }{4\Big(\ga  +
  i\al\Big)}  ,
\end{align*}
where $\al=\pi m/\theta$. Hence,
\begin{align*}
  0=8\al i
  \Big\{\ga  (c+c') -2\ga \al^2 +i[\al(c+c')-\al^3+\ga^2\al] \Big\}+
    (c-c')^2.
\end{align*}
This leads to
$$
 c+c' =2\al^2
$$
and
\begin{align*}
  0=-8\al^2
  (c+c'-\al^2+\ga^2) +
    (c-c')^2.
\end{align*}
The second equation yields
\begin{align*}
  8\al^2
  ( \al^2+\ga^2) =
    (c-c')^2,
\end{align*}
which implies that for $\al\not=0$ we have
\begin{align*}
  |c-c'|=2\sqrt{2}\{\al^2
  ( \al^2+\ga^2)\}^{1/2} >2\sqrt{2}\al^2=\sqrt{2}(c+c'),
\end{align*}
which is a contradiction.
This ends the proof of \eqref{040505-22a} and therefore the
demonstration of  \eqref{021504-22}.}
\qed

\subsection{Proof of Lemma  \ref{lm021504-22}}

Using \eqref{011404-22}  we obtain
\begin{align}
\label{081404-22}
\sum_{x=0}^n|M_{x,0}(m)|^2= \sum_{j_1,\ldots,j_4 =0}^n\Theta_m(\mu_{j_1},\mu_{j_2}) \Theta_m(\mu_{j_3},\mu_{j_4}) \prod_{k=1}^4\psi_{j_k}(0)   \sum_{x=0}^n\prod_{k=1}^4\psi_{j_k}(x) .
\end{align}
Applying elementary trigonometric identities we conclude that
\begin{align*}
&\sum_{x=0}^n\prod_{k=1}^4\psi_{j_k}(x) =\frac{1}{(n+1)^2} 
  \left\{\prod_{k=1}^4(2-\delta_{0,j_k})\right\}^{1/2} \sum_{x=0}^n\prod^4_{k=1}\cos\left(\frac{\pi
    j_k(2x+1)}{2(n+1)}\right) \\
&
=\frac{1}{2^5(n+1)} 
  \left\{\prod_{k=1}^4(2-\delta_{0,j_k})\right\}^{1/2} \sum_{\iota_1,\ldots,\iota_4\in\{-1,1\}}\cos\left(\frac{\pi
     }{2(n+1)}\sum_{k=1}^4\iota_kj_k\right) 1_{2(n+1)\bbZ}\Big(\sum_{k=1}^4\iota_kj_k\Big)
\end{align*}
Therefore we can write
\begin{align}
\label{091404-22}
&\sum_{x=0}^n|M_{x,0}(m)|^2= \sum_{\iota,\iota_1\in\{-1,1\}}\sum_{\iota',\iota_1'\in\{-1,1\}}\int_0^1\dd u\int_0^1\dd u' \int_0^1\dd
  u_1\int_0^1\dd u'_1 {\frak V}_{m}(u,u'){\frak V}_{m}^\star(u_1,u'_1)\\
&
\times 
  \exp\left\{\ii\pi\Big( \iota
  u+\iota'u'+\iota_1
  u_1+\iota'_1u'_1\Big)\right\}\sum_{q\in\bbZ}\delta_q\Big(\iota
  u+\iota'u'+\iota_1
  u_1+\iota'_1u'_1\Big)+O_m\left(\frac1n\right),
\end{align}
where
$O_m\left(\frac1n\right)\le \frac{C}{n}$ for some constant $C>0$,
independent of $n$ and $m$, and
\begin{align*}
{\frak V}_{m}(u,u'):=  \Theta_m\Big(\om_0^2+4\sin^2\left(\frac{\pi u}{2}\right), \om_0^2+4\sin^2\left(\frac{\pi u'}{2}\right)\Big)\cos\left(\frac{\pi u}{2}\right)\cos\left(\frac{\pi u'}{2}\right).
\end{align*}
We claim that   there exists $C>0$, independent of $n$ and $m$, such that ${\frak V}_{m}(u,u')\le
C$ for all $u,u'\in[0,\om_0^2+4]$. Indeed,  as can be seen directly from \eqref{Thetam},  we have
$\lim_{m\to+\infty}\Theta_m(c,c')=0$ uniformly in
$c,c'\in[0,\om_0^2+4]$. On the other hand  the function $\bbR\times
[0,\om_0^2+4]^2\ni(m,c,c')\to \Theta_m(c,c')$ is  bounded on
compact set. If otherwise, this would imply that there exist
$(m,c,c')\in \bbR\times [0,\om_0^2+4]^2$ such that
\begin{align*}
  &0=\Big[ \Big(2\ga  +\frac{\pi i
  m}{\theta}\Big)\Big(c+c' +2 \Big(\ga  +\frac{\pi i
  m}{\theta}\Big) \frac{\pi i
  m}{\theta}\Big) \Big]^{-1} \\
&
\times \Big[\frac14\Big(\ga  +\frac{\pi i
  m}{\theta}\Big)^{-1}(c-c')^2 + \frac{\pi i
  m}{\theta} (c+c') \Big]
                                  +1.
\end{align*}
An easy calculation gives $c+c'=\ga^2-2\al^2$ and
$
(c-c')^2=8(\al^2+\ga^2)^2,$
where $\al=\pi m/\theta$. This leads to
a contradiction, as then $|c-c'|>c+c'$ (but both $c,c'>0$).
Thus the conclusion of the lemma
follows. \qed

\subsection{Proof of Lemma  \ref{lm011904-22}}

{From \eqref{022212-21} we obtain
\begin{align*}
[e^{-At}]_{x+n+1,x'+n+1}=
    \sum_{j=0}^nE_j(t)\psi_j(x)\psi_j(x'),
\end{align*}
where (cf \eqref{muj})
$$
E_j(t):=\frac{1}{2\sqrt{\ga^2-\mu_j}}\Big[
  -\la_{j,+} \exp\left\{\la_{j,+}t\right\}+
    \la_{j,-} 
    \exp\left\{\la_{j,-}t\right\}\Big],\quad \mbox{if }\mu_j\not= \ga^2.
$$
In the case $\mu_j= \ga^2$ (then $\la_{j,\pm}=\ga$, cf \eqref{lapm})
we have
$
E_j(t):=(1-\ga t)e^{-\ga t} .
$
Using \eqref{010501-22a} we obtain therefore 
\begin{align}
\label{021904-22}
&\bar p_x(t)= \frac{1}{n^{1/2}} \sum_{j=0}^n\int_{0}^{+\infty} \cF((t-s)/\theta )
[e^{-A s}]_{x+n+1,2n+1} \dd s \\
&
=\frac{1}{n^{1/2}}  \sum_{j=0}^n\psi_j(x)\psi_j(n)
  \int_{0}^{+\infty} \cF((t-s)/\theta )E_j(s)\dd s.\notag
\end{align}}
From \eqref{021904-22} we conclude that there exists ${\frak p}_*>0$
such that
\begin{equation}
\label{031904-22}
\sup_{t\in\bbR,x=0,\ldots,n}|\bar p_x(t)|\le \frac{{\frak p}_*}{n^{1/2}},\quad n=1,2,\ldots.
\end{equation}
Estimate \eqref{011904-22} is then a straightforward consequence of
\eqref{031904-22} and \eqref{021012-21}.\qed

\subsection{Proof of Lemma  \ref{lm031504-22}}

Multiplying both sides of  \eqref{eq:92ab} by $V_x(t)$ and averaging
over time we get
\begin{equation}
    \label{eq:92abc}
\begin{split}
&
\lang V_x^2\rang =  \sum_{x'=1}^n \int_0^\theta
    {\frak g}_{x,x'}(s) \lang V_x(\cdot)V_{x'}(\cdot-s)\rang \dd s+
    \lang V_x v_x\rang \\
&
\le  \sum_{x'=1}^n  M_{x,x'} \lang V_x^2\rang^{1/2}\lang V_{x'}^2\rang^{1/2}+
    \lang V_x v_x\rang.
\end{split}
  \end{equation}
Summing up over $x$ we obtain
\begin{align*}
&\sum_{x,x'=0}^n \Big(\delta_{x,x'} -M_{x,x'}\Big) \lang
  V_x^2\rang^{1/2}\lang V_{x'}^2\rang^{1/2}+\lang V_0^2\rang^{1/2}\sum_{x=0}^n  M_{x,0}
  \lang V_{x}^2\rang^{1/2}\\
&
\le \sum_{x=0}^n \lang V_x v_x\rang.
\end{align*}
Using \eqref{lowerM} and  the Cauchy-Schwarz inequality we obtain in
particular that
\begin{align*}
M_{0,0}\lang V_0^2\rang
\le \left\{\sum_{x=0}^n \lang V_x^2
  \rang\right\}^{1/2}\left\{\sum_{x=0}^n \lang v_x^2
  \rang\right\}^{1/2}\le \frac{C}{n}\left\{\sum_{x=0}^n \lang V_x \rang^2\right\}^{1/2}
\end{align*}
for some $C>0$ independent of $n$.
The last estimate follows from \eqref{011904-22}.
To finish the proof note that from \eqref{eq:54} we have
\begin{align*}
&M_{0,0}
  =\frac{4}{(n+1)^2}\sum_{j,j'=0}^n(1-\delta_{0,j})(1-\delta_{0,j'})\Theta(\mu_j,\mu_{j'})\cos^2\left(\frac{\pi
  j}{2(n+1)}\right) \cos^2\left(\frac{\pi j'}{2(n+1)}\right)\\
&
\approx 4\gamma^2\int_0^1\int_0^1
  \frac{\Big[\om_0^2+2\sin^2(\pi u/2)+2\sin^2(\pi u'/2)\Big]\cos^2(\pi
  u/2)\cos^2(\pi u'/2)\dd u\dd u'}
{\gamma^2 \Big(\om_0^2+2\sin^2(\pi u/2)+2\sin^2(\pi u'/2) \Big)+ \Big(\sin^2(\pi u/2)-\sin^2(\pi u'/2))\Big)^2},
\end{align*}
where the equality holds up to a term of order $O(1/n)$.\qed

\section{Calculation of ${\cal Q}^{a,b}(\ell)$}
 
\label{appD}

 Consider the
Green's function of $-\Delta + \omega^2_0$ given by \eqref{GR}. We can write
\begin{align}
\label{GRC0}
G_{\om_0^2}(x) 
    =\lim_{\eta\to0+}\frac{1}{\la\sqrt{1+4/\la}}\left\{1+\frac{\la}{2}\Big(1+\sqrt{1
                                                                       +\frac{4}{\la}}\Big)\right\}^{-|x|}_{|\la=\om_0^2+i\eta},\quad
  x\in\bbZ. 
\end{align}
Here the square root denotes the branch of the inverse of $z\mapsto
z^2$ such that ${\rm Re}\,\sqrt{\la}>0$, when $\la\in\mathbb
C\setminus(-\infty,0]$. By the analytic continuation we have therefore
\begin{align}
\label{GRC}
 & G_{\la}(x) = \left(-\Delta + \la
   \right)^{-1}(x)=\int_0^1\frac{\cos(2\pi ux)du}{4\sin^2(\pi
   u)+\la }\notag\\
  &
    = \frac{1}{\la\sqrt{1+4/\la}}\left\{1+\frac{\la}{2}\Big(1+\sqrt{1
                                                                       +\frac{4}{\la}}\Big)\right\}^{-|x|}
    ,\quad x\in\bbZ 
\end{align}
for $\la\in\mathbb
C\setminus[-4,0]$.

Recall that  ${\cal Q}^{a,b}(\ell)$ is given by 
 \eqref{021205-21f1} and  \eqref{021205-21f2} with $a+b=1/2$. Let
\begin{equation}
  \label{021205-21f1t}
\begin{aligned}
  & \tilde{\cal Q}^{-1/2,0}(\ell)= 4\gamma|\tilde     \cF(\ell)|^2
  \left(\frac {2\pi\ell }{\theta}\right)^2  \int_0^1 \cos^2\left(\frac{\pi z}{2}\right)
  \left\{\left[4\sin^2\left(\frac{\pi z}{2}\right)
      +\om_0^2 -\left(\frac{2\pi\ell}{\theta}\right)^2\right]^2
    +\left(\frac{\red{4} \gamma \pi \ell}{\theta}\right)^2 \right\}^{-1}\dd z
\end{aligned} \end{equation}
and 
\begin{equation}
  \label{021205-21f1t}
\begin{aligned}
   \tilde{\cal Q}^{b-1/2,b}(\ell)= 4\gamma|\tilde     \cF(\ell)|^2 \left(\frac {2\pi\ell }{\theta}\right)^2 
   \int_0^1 \cos^2\left(\frac{\pi z}{2}\right)
   \left[4\sin^2\left(\frac{\pi z}{2}\right)
       +\om_0^2 
     \right]^{-2}\dd z
,\quad\mbox{when }b>0.
\end{aligned} 
\end{equation}
Let  
$$
\la(\om_0):=\om_0^2 -\left(\frac{2\pi\ell}{\theta}\right)^2 
    +i\left(\frac{{4} \gamma \pi \ell}{\theta}\right) .
$$ We have
\begin{align*}
  & \tilde{\cal Q}^{-1/2,0}(\ell)  
    =-\frac{2\pi\ell|\tilde     \cF(\ell)|^2}{\theta}
 {\rm Im}\left(\int_0^{1} \frac{[1+\cos\left(2\pi z\right)] \dd z}{
 4\sin^2\left(\pi z\right)
      +\la(\om_0) }\right)
\end{align*}
    Using \eqref{GRC} we can write
\begin{align*}
  & \tilde{\cal Q}^{-1/2,0}(\ell) 
    =-\frac{2\pi\ell|\tilde     \cF(\ell)|^2}{\theta}
    {\rm Im}\left(G_{\la(\om_0)}(0)+G_{\la(\om_0)}(1)\right)\\
   &
     =-\frac{2\pi\ell|\tilde     \cF(\ell)|^2}{\theta}
    {\rm Im}\left( \left\{\frac{2 }{\la(\om_0)\sqrt{1+4/\la(\om_0)}}+\frac12\right\} \left\{1+\frac{ \la(\om_0)}{2}\Big(1+\sqrt{1
     +\frac{4}{\la(\om_0)}}\Big)\right\}^{-1}\right)
\end{align*}
Next, for $b>0$, we have
\begin{align*}
  & \tilde{\cal Q}^{b-1/2,b}(\ell) 
    = 4\gamma|\tilde     \cF(\ell)|^2 \left(\frac {2\pi\ell }{\theta}\right)^2 
   \int_0^{1} \frac{\cos^2\left(\pi z\right) \dd z}{
   \left[4\sin^2\left(\pi z\right)
       +\om_0^2 
    \right]^{2}}\\
     &
    = 2\gamma|\tilde     \cF(\ell)|^2 \left(\frac {2\pi\ell }{\theta}\right)^2 
   \int_0^{1} \frac{[1+\cos\left(2\pi z\right)] \dd z}{
   \left[2 +\om_0^2 -2\cos\left(2\pi z\right)    
     \right]^{2}}
    = \gamma|\tilde     \cF(\ell)|^2   \left(\frac {2\pi\ell
     }{\theta}\right)^2{\cal I},
\end{align*}
where
\begin{equation}
  \label{cI}
  {\cal I}:=\frac{1}{2\pi i}
      \int_{C}f(\zeta) \dd \zeta
      \end{equation}
is the integral over the  circle $C:=[|\zeta|=1]$, oriented counter-clockwise and
$$
f(\zeta):=\left(\frac{\zeta+1  }{
   1+(2 +\om_0^2) \zeta -\zeta^2 
      }\right)^2
$$
The poles of $f(\zeta)$ occur  
  at
$$
\zeta_\pm:=\frac12\Big[2 +\om_0^2\pm\sqrt{\om_0^4+4 \om_0^2+8} \Big].
$$
We have $|\zeta_-|<1< \zeta_+$ and 
$$
f(\zeta):=\left(\frac{A_-}{\zeta-\zeta_-}+\frac{A_+}{\zeta-\zeta_+}\right)^2,
$$
with
$$
A_-:=\frac{\zeta_-+1}{\zeta_--\zeta_+},\qquad A_+:=\frac{\zeta_++1}{\zeta_+-\zeta_-}.
$$
Hence, after performing the contour integration, we get
\begin{align*}
{\cal I}=2A_-A_+\left( \frac{1}{2\pi
                 i}\int_{C}\frac{1}{\zeta-\zeta_+}\cdot \frac{ \dd
                 \zeta}{\zeta-\zeta_-}\right)=\frac{2A_-A_+}{\zeta_--\zeta_+}=\frac{2(\zeta_-+1)
                 (\zeta_++1)}{(\zeta_+-\zeta_-)^3}
                 \end{align*}
  Since
  \begin{align*}
    \zeta_-\zeta_+=1,\quad \zeta_-+\zeta_+=2 +\om_0^2,\quad
      \zeta_+-\zeta_-=\sqrt{\om_0^4+4 \om_0^2+8}
  \end{align*}
  we obtain
   $$
{\cal I}=\frac{2(4 +\om_0^2) }
    { (\om_0^4+4 \om_0^2+8)^{3/2}} .
   $$
   and
   \begin{align*}
   \tilde{\cal Q}^{b-1/2,b}(\ell)
    =\left(\frac {2\pi\ell }{\theta}\right)^2 
    \frac{ 2\gamma|\tilde     \cF(\ell)|^2  (4 +\om_0^2) }
    { (\om_0^4+4 \om_0^2+8)^{3/2}} .
\end{align*}

\bibliographystyle{amsalpha}

\begin{thebibliography}{A}

\bibitem{BOS10}   Basile, G.; Olla, S.; Spohn, H. {\em  Energy
    transport in stochastically perturbed lattice
    dynamics}. Arch. Ration. Mech. Anal. 195 (2010), no. 1, 171–203.

   \bibitem{BKLL11}    Bernardin, C., Kannan, V., Lebowitz, J.L., Lukkarinen, J.:
    \emph{Nonequilibrium stationary states of harmonic chains with bulk noises}.
    Eur. Phys. J. B 84, 685–689 (2011)

\bibitem{bo1} C. Bernardin, S. Olla,  {\em Transport Properties of a Chain
  of Anharmonic Oscillators with Random Flip of Velocities},  J. Stat
Phys (2011) 145:1224-1255 DOI 10.1007/s10955-011-0385-6


\bibitem{BRV70} M. Bolsterli, M. Rich, and W. M. Visscher, {\em Simulation of nonharmonic interactions in a crystal by self-consistent reservoirs,} Phys. Rev. A 4:1086–1088 (1970).





\bibitem{bll} F. Bonetto, J. L. Lebowitz and J.
  Lukkarinen  {\em Fourier’s Law for a Harmonic Crystal with Self-Consistent Stochastic Reservoirs.}
J. of Stat. Physics, Vol. 116,  2004  

\bibitem{GS79}  Gihman, I. I.; Skorohod, A. V. {\em The theory of stochastic processes. III.} Translated from the Russian by Samuel Kotz. With an appendix containing corrections to Volumes I and II. Grundlehren der Mathematischen Wissenschaften, 232. Springer-Verlag, Berlin-New York, 1979. 


  \bibitem{jko15}  Jara, M.; Komorowski, T.; Olla, S. {\em Superdiffusion of energy in a chain of harmonic oscillators with noise.} Comm. Math. Phys. 339 (2015), no. 2, 407–453.



   






  \bibitem{kasminski} Rafail Khasminskii,
    \emph{Stochastic Stability of Differential Equations},
    Stochastic Modelling and Applied Probability 66, Springer.
    
  


 \bibitem{klo2} T. Komorowski, J. L., Lebowitz, S. Olla, {\em Heat flow in
     a periodically forced, thermostatted chain II}, (2022) https://arxiv.org/abs/2209.12923,
   to appear in J.Stat.Phys.
    
    \bibitem{klos} T. Komorowski, J. L., Lebowitz, S. Olla, M. Simon,
   {\em On the Conversion of Work into Heat: Microscopic Models and Macroscopic Equations},
   (2022), https://arxiv.org/abs/2212.00093, to appear
   in Ensaios Matem\'aticos for the volume dedicated to Errico Presutti 80th birthday.

\bibitem{kos2}  T. Komorowski, S. Olla, M. Simon,
  \emph{An open microscopic model of heat conduction:
    evolution and non-equilibrium stationary states},
  Communications in Mathematical Sciences,
  18, vol. 3, 751-780, (2020).
  https://dx.doi.org/10.4310/CMS.2020.v18.n3.a8

  


\bibitem{ray}  Koushik R., {\em Green’s function on lattices}, available at
  {\tt https://arxiv.org/abs/1409.7806}


\bibitem{joel57}
  J.L. Lebowitz, P.G. Bergmann, \emph{Irreversible Gibbsian Ensembles},
  Annals of Physics, Vol. 1, N.1, 1-23,
  1957. https://doi.org/10.1016/0003-4916(57)90002-7


  \bibitem{lukk} Lukkarinen, J.,
  \emph{Thermalization in Harmonic Particle Chains with Velocity Flips}.
  J Stat Phys 155, 1143–1177 (2014). https://doi.org/10.1007/s10955-014-0930-1


 \bibitem{BRV75}  M. Rich and W. M. Visscher, {\em Disordered harmonic
    chain with self-consistent reservoirs,}  Phys. Rev. B 11:2164–2170 (1975).

\bibitem{RLL67}   Rieder, Z., Lebowitz, J.L., Lieb, E.: Properties of
  harmonic crystal in a stationary non-equilibrium state. J. Math. Phys. 8, 1073–1078 (1967)
\end{thebibliography}

\end{document}